\documentclass[screen]{lmcs}
\pdfoutput=1

\usepackage[utf8]{inputenc}
\usepackage{lastpage}
\lmcsdoi{18}{1}{22}
\lmcsheading{}{\pageref{LastPage}}{}{}%
{Dec.~01,~2020}{Feb.~01,~2022}{}


\usepackage{color}

\definecolor{BlueViolet}{rgb}{0, 0, 0.55}
\definecolor{RubineRed}{rgb}{0.88, 0.07, 0.37}
\definecolor{ForestGreen}{rgb}{0.13, 0.55, 0.13}
\definecolor{Blue}{rgb}{0.0, 0.0, 1.0}
\definecolor{NavyBlue}{rgb}{0.0, 0.0, 0.5}
\definecolor{Black}{rgb}{0.02, 0.02, 0.02}
\definecolor{MidnightBlue}{rgb}{0.0, 0.2, 0.4}
\definecolor{Gray}{rgb}{0.41, 0.41, 0.41}
\definecolor{TealBlue}{rgb}{0.212,0.459,0.533}
\definecolor{Plum}{rgb}{0.6,0.25,0.6}

\usepackage{amssymb}
\usepackage{tikz}

\usetikzlibrary{shadows}
\usetikzlibrary{chains,positioning,decorations.pathreplacing}
\usetikzlibrary{patterns,fit,arrows}
\usetikzlibrary{shapes.geometric,shapes.misc}

\tikzset{
  block/.style={
  draw, 
  rectangle, 
  minimum height=1.5cm, 
  minimum width=2.5cm, 
  align=center,
  draw=red!50!black!50, 
  top color=white, 
  bottom color=red!50!black!20, 
  }, 
  line/.style={->,>=latex'}
}

\tikzset{
  inputblock/.style={
  draw, 
  rounded rectangle, 
  minimum height=1.2cm, 
  minimum width=2.5cm, 
  align=center,
  very thick,draw=black!50,
  top color=white,bottom color=black!20,
  }, 
  line/.style={->,>=latex'}
}

\tikzset{
  outputblock/.style={
  draw, 
  rounded rectangle, 
  minimum height=1.2cm, 
  minimum width=2.5cm, 
  align=center,
  very thick,draw=black!50,
  top color=white,bottom color=blue!20,
  }, 
  line/.style={->,>=latex'}
}

\usepackage{mathtools}
\usepackage{url}
\usepackage{etoolbox}
\usepackage{xspace} 
\usepackage{stmaryrd} 
\usepackage{booktabs}

\usepackage[inline,shortlabels]{enumitem} 
\newlist{inlinelist}{enumerate*}{1}
\setlist*[inlinelist,1]{%
  label=(\roman*),
}

\usepackage{xifthen}  
\usepackage{ifthen}

\usepackage{nicefrac}
\usepackage{hyperref}
\usepackage{cleveref}

\usepackage[final,nomargin,inline,index]{fixme} 
\fxusetheme{color}

\FXRegisterAuthor{bart}{anbart}{\color{magenta} {\underline{bart}}}

\hypersetup{
  breaklinks   = true,
  colorlinks   = true, 
  urlcolor     = blue, 
  linkcolor    = blue, 
  citecolor    = red   
}
\hypersetup{final} 

\usepackage{listings}

\definecolor{listingBG}{HTML}{FFFFCB}%
\definecolor{listingFrame}{HTML}{BBBB98}%
\definecolor{listingLineno}{rgb}{0.5,0.5,1.0}%

\definecolor{LightGrey}{rgb}{0.975,0.975,0.975}

\lstdefinelanguage{bitml}{
	commentstyle=\color{Gray},
	morecomment=[l]{;},
	morecomment=[s]{/*}{*/},
	classoffset=0,
        escapechar=\$,
	morekeywords={contract,pre,choice,reveal,split,put,revealif,withdraw,after},
	keywordstyle=\color{Blue}\bfseries,
	classoffset=1,
	morekeywords={deposit,secret,participant,sig,versig,fun,unit,int,string,bool,address,uint},
	keywordstyle=\color{TealBlue},
	classoffset=2,
	morekeywords={BTC,true},
	keywordstyle=\color{Plum}\bfseries,
}

\newcommand{\ifempty}[3]{%
  \ifthenelse{\isempty{#1}}{#2}{#3}%
}

\newcommand{\ifdots}[3]{%
  \ifthenelse{\equal{#1}{...}}{#2}{#3}%
}

\newcommand{\mypar}[1]{\paragraph*{#1}}

\newcommand{\Real}[1]{\mathrm{Real}}

\newcommand{\codefont}{\fontsize{10}{10}\selectfont}
\newcommand{\code}[1]{{\tt\codefont {#1}}}

\newcommand{\eg}{e.g.\@\xspace}
\newcommand{\ie}{i.e.\@\xspace}
\newcommand{\wrt}{w.r.t.\@\xspace}

\newcommand{\subst}[2]{\{{#1}/{#2}\}}

  {}

\newcommand{\BTC}{\textup{%
  \leavevmode
  \vtop{\offinterlineskip 
    \setbox0=\hbox{B}%
    \setbox2=\hbox to\wd0{\hfil\hskip-.03em
    \vrule height .3ex width .15ex\hskip .08em
    \vrule height .3ex width .15ex\hfil}
    \vbox{\copy2\box0}\box2}}\xspace}

\def\pmvColor{\color{ForestGreen}}
\newcommand{\pmvFmt}[1]{{\pmvColor{\sf #1}}}

\newcommand{\PartT}{\pmvFmt{{Hon}}\xspace} 

\newcommand{\pmv}[2][]{\pmvFmt{#2}_{\pmvColor{#1}}\xspace}
\newcommand{\pmvA}[1][]{\pmv[{#1}]{A}}
\newcommand{\pmvB}[1][]{\pmv[{#1}]{B}}
\newcommand{\pmvC}[1][]{\pmv[{#1}]{C}}

\newcommand{\pmvM}[1][]{\pmv[{#1}]{M}} 
\newcommand{\pmvP}[1][]{\pmv[{#1}]{P}}

\newcommand{\concrightarrow}[1]{\ifempty{#1}{\rightarrow_{\flat}}{\xrightarrow{#1}_{\flat}}}
\newcommand{\absrightarrow}[1]{\ifempty{#1}{\rightarrow_{\sharp}}{\xrightarrow{#1}_{\sharp}}}
\newcommand{\finrightarrow}[1]{\ifempty{#1}{\rightarrow_{\sharp\textit{fin}}}{\xrightarrow{#1}_{\sharp\textit{fin}}}}

\newcommand{\txIn}[2][]{{\small\textsf{in}}\ifempty{#1}{\ifempty{#2}{}{: {#2}}}{({#1})\ifempty{#2}{}{: {#2}}}}

\newcommand{\timeT}{t}

\DeclareMathAlphabet{\mathbfsf}{\encodingdefault}{\sfdefault}{bx}{n}

\newcommand{\mmid}{\,\|\,}

\newcommand{\cn}[1]{\ifempty{#1}{\operatorname{cn}}{\operatorname{cn}(#1)}}

\newcommand{\irule}[2]{\dfrac{#1}{#2}}

\newcommand{\bnfdef}{::=}
\newcommand{\bnfmid}{\;|\;}

\newcommand{\nrule}[1]{{\scriptsize \textsc{#1}}}
\newcommand{\smallnrule}[1]{{\tiny \textsc{#1}}}
\newcommand{\sem}[2][]{\mbox{\ensuremath{\llbracket{#2}\rrbracket_{#1}}}}

\newcommand{\Nat}{\mathbb{N}}

\newcommand{\bind}[2]{\nicefrac{#2}{#1}}
\newcommand{\setenum}[1]{\{#1\}}
\newcommand{\setcomp}[2]{\left\{{#1} \,\middle|\, {#2}\right\}}

\newcommand{\qedef}{\ensuremath{\diamond}}

\definecolor{LightGrey}{rgb}{0.95,0.95,0.95}
\definecolor{keyword}{HTML}{7F0055}

\newlength\replength
\newcommand\repfrac{.1}

\newcommand\rulewidth{.6pt}
\newcommand\tdashfill[1][\repfrac]{\cleaders\hbox to \replength{%
  \smash{\rule[\arraystretch\ht\strutbox]{\repfrac\replength}{\rulewidth}}}\hfill}

\newcommand\tdotfill[1][\repfrac]{\cleaders\hbox to \replength{%
  \smash{\raisebox{\arraystretch\dimexpr\ht\strutbox-.1ex\relax}{.}}}\hfill}

\newcommand{\var}[2][]{#2_{#1}} 
\newcommand{\varX}[1][]{\var[#1]{x}} 
\renewcommand{\varXi}[1][]{\var[#1]{x'}} 
 
\newcommand{\varY}[1][]{\var[#1]{y}} 
 
\newcommand{\varZ}[1][]{\var[#1]{z}}
 
\newcommand{\VarX}[1][]{\var[#1]{X}}

\newcommand{\varph}[2][]{#2_{#1}} 
\newcommand{\varphX}[1][]{\varph[#1]{d}} 
\newcommand{\varphXi}[1][]{\varph[#1]{d'}} 
\newcommand{\varphY}[1][]{\varph[#1]{e}} 
\newcommand{\varphYi}[1][]{\varph[#1]{e'}} 

\newcommand{\dmv}[1][]{\chi_{#1}}
\newcommand{\dmvi}[1][]{\chi'_{#1}}

\newcommand{\authmv}[1][]{\xi_{#1}}

\newcommand{\secr}[2][]{\mathord{{#2}_{#1}}}
\newcommand{\presecret}[2]{{#1}\!:\!\textup{\texttt{secret}}\,{\secr{#2}}}
\newcommand{\secrA}[1][]{\secr[#1]{a}} 
\newcommand{\secrAi}[1][]{\secr[#1]{a'}} 
\newcommand{\secrB}[1][]{\secr[#1]{b}}

\newcommand{\secrOf}[1]{\mathcal{S}_{#1}}

\newcommand{\const}[2][]{#2_{#1}} 
\newcommand{\constK}[1][]{\const[#1]{k}}

\newcommand{\constT}[1][]{\const[#1]{t}}
\newcommand{\constTi}[1][]{\const[#1]{t'}}

\newcommand{\val}[2][]{#2_{#1}} 
\newcommand{\valV}[1][]{\val[#1]{v}}
\newcommand{\valVi}[1][]{\val[#1]{v'}}

\newcommand{\valW}[1][]{\val[#1]{w}}

\newcommand{\true}{\mathit{true}}

\def\contrColor{\color{RubineRed}}

\newcommand{\contrFmt}[1]{{\contrColor{#1}}}
\newcommand{\contrC}[1][]{\mathord{\contrFmt{C}_{\contrColor{#1}}}}
\newcommand{\contrCi}[1][]{\mathord{\contrC[#1]\contrColor{'}}}
\newcommand{\contrCii}[1][]{\mathord{\contrC[#1]\contrColor{''}}}

\newcommand{\contrD}[1][]{\mathord{\contrFmt{D}_{\contrColor{#1}}}}
\newcommand{\contrDi}[1][]{\mathord{\contrD[#1]\contrColor{'}}}

\newcommand{\expE}[1][]{\ensuremath{E}_{#1}}

\newcommand{\contrG}[1][]{\mathord{\contrFmt{G}_{\contrColor{#1}}}}
\newcommand{\contrGi}[1][]{\mathord{\contrG[#1]\contrColor{'}}}

\newcommand{\predP}[1][]{\mathord{p_{#1}}}

\newcommand{\contrAdv}[3][]{\setenum{#2}^{#1}{#3}}  
\newcommand{\contrAdvC}[2]{\mathcal{C}} 

\newcommand{\persdep}[3]{\mbox{\ensuremath{{#1}\textup{:}\,{#2}\,\textup{\texttt{@}}\,{#3}}}}

\newcommand{\assign}[3]{{#1}:{#2}\leftarrow{#3}}

\newcommand{\putname}{\textup{\texttt{put}}}
\newcommand{\revealname}{\textup{\texttt{reveal}}}
\newcommand{\rngtname}{\textup{\texttt{rngt}}}
\newcommand{\wherename}{\textup{\texttt{if}}}
\newcommand{\andputname}{\&}
\newcommand{\putCnoreveal}[1]{%
  \ifempty{#1}{\tau}{\putname \, {#1}}
}

\newcommand{\putCtrue}[2]{%
  \ifempty{#2}
  {\putCnoreveal{#1}}
  {\ifempty{#1}
    {\revealname \, {#2}}
    {\putname \, {#1} \, \revealname \, {#2}}}
}

\newcommand{\putC}[3][]{
  \ifempty{#1}
  {\putCtrue{#2}{#3}}
  {\ifempty{#2}
    {\revealname \, {#3} \, \wherename \, {#1}}
    {\putname \, {#2} \, \andputname \, \revealname \, {#3} \, \wherename \, {#1}}
  }
}

\newcommand{\cnil}{0}

\newcommand{\splitname}{\textup{\texttt{split}}}
\newcommand{\splitC}[1]{\splitname \ifempty{#1}{}{\; #1}}
\newcommand{\splitB}[2]{\ifempty{#1}{}{{#1} \rightarrow }{#2}}

\newcommand{\withdrawname}{\textup{\texttt{withdraw}}}
\newcommand{\withdrawC}[1]{\withdrawname\ifempty{#1}{}{\; {\pmv{#1}}}}

\newcommand{\authC}[2]{{\pmv{#1}}\,\textup{:}\,{#2}}

\newcommand{\aftername}{\texttt{after}}
\newcommand{\afterC}[2]{\textup{\aftername}\,{#1}\,\textup{:}\,{#2}}

\newcommand{\cVar}[2][]{\texttt{\textup{#2}}_{#1}}
\newcommand{\cVarX}[1][]{\cVar[#1]{X}}
\newcommand{\cVarY}[1][]{\cVar[#1]{Y}}

\newcommand{\procParamA}[1][]{\beta_{#1}}
\newcommand{\procParamB}[1][]{\beta'_{#1}}

\newcommand{\sexp}[1][]{\mathcal{E}_{#1}}
\newcommand{\sexpi}[1][]{\mathcal{E}'_{#1}}

\newcommand{\decl}[2]{{#1}({#2})}

\newcommand{\call}[2]{{#1}\ifempty{#2}{}{\langle{#2}\rangle}}
\newcommand{\callX}[1]{\call{\cVarX}{#1}}
\newcommand{\callY}[1]{\call{\cVarY}{#1}}

\newcommand{\advPref}[1]{{\pmvColor *}:{#1}}
\newcommand{\adv}[1]{\advPref{\rngtname}\ifempty{#1}{}{\;}{#1}}
\newcommand{\extadv}[2]{{#1}:\rngtname\ifempty{#1}{}{\;}{#2}}
\newcommand{\ncadv}[1]{\textup{\texttt{call}}\ifempty{#1}{}{\;}{#1}}

\newcommand{\confG}[1][]{\Gamma_{#1}}
\newcommand{\confGi}[1][]{\Gamma'_{#1}}

\newcommand{\confD}[1][]{\Delta_{#1}}
\newcommand{\confDi}[1][]{\Delta'_{#1}}

\newcommand{\gnil}{0}
\newcommand{\confContr}[3][]{\langle {#2}\ifempty{#3}{}{, {#3}} \rangle_{#1}}
\newcommand{\confDep}[3][]{\langle {#2}, {#3} \rangle_{#1}}
\newcommand{\confAuth}[2]{{#1}[{#2}]}
\newcommand{\confRev}[3]{{#1} : {#2}\#{#3}}
\newcommand{\confSec}[3]{\setenum{\confRev{#1}{#2}{#3}}}

\newcommand{\confTsep}{\mid}
\newcommand{\confT}[2]{{#1} \confTsep {#2}}

\newcommand{\authLab}[2]{{#1}:{\ifdots{#2}{\cdots}{#2}}} 

\newcommand{\authCommit}[3][]{\textbf{\#} \rhd \contrAdv[#1]{#2}{#3}}
\newcommand{\authAdv}[4][]{{#2} \rhd \contrAdv[#1]{#3}{#4}}
\newcommand{\authBranch}[2]{{#1} \rhd {#2}}
\newcommand{\authJoin}[4]{{#1},{#2} \rhd \confDep{#3}{#4}}
\newcommand{\authSplit}[4]{{#1} \rhd \confDep{#2}{#3},\confDep{#2}{#4}}
\newcommand{\authDestroy}[3]{{#1},{#2} \rhd {#3}}
\newcommand{\authDonate}[2]{{#1} \rhd {#2}}

\newcommand{\abs}[2][]{\ifempty{#2}{\alpha_{#1}}{\alpha_{#1}{({#2})}}}
\newcommand{\absConf}[2][]{\abs[{#1}]{#2}}
\newcommand{\absContr}[2][]{\abs[{#1}]{#2}}

\newcommand{\absConfG}[1][]{\confG[{#1}]^{\sharp}}
\newcommand{\absConfGi}[1][]{\confGi[{#1}]{\hspace{1pt}}^{\sharp}}

\newcommand{\absRunname}{\mathcal{R}}
\newcommand{\absRunS}[1][]{\absRunname^{\sharp}_{#1}}

\newcommand{\origin}[3][]{\mathit{orig}_{#1}\ifempty{#2}{}{(#2,#3)}}
\newcommand{\descen}[3][]{\mathit{desc}_{#1}\ifempty{#2}{}{(#2,#3)}}

\newcommand{\labS}[1][]{\ell_{#1}}  

\newcommand{\absLabS}[1][]{\ell_{#1}}  

\newcommand{\confS}[1]{\confG[{#1}]}

\newcommand{\runnameS}{\mathcal{R}}
\newcommand{\runS}[1][]{\runnameS_{#1}}

\newcommand{\ZCB}{\ensuremath{\contrFmt{\it ZCB}}\xspace}
\newcommand{\ZCBi}{\ensuremath{\contrFmt{\it ZCB2}}\xspace}
\newcommand{\ZCBii}{\ensuremath{\contrFmt{\it ZCB3}}\xspace}

\newcommand{\Escrow}{\ensuremath{\contrFmt{\it Escrow}}\xspace}
\newcommand{\Resolve}{\ensuremath{\contrFmt{\it Resolve}}\xspace}

\newcommand{\Refund}[1][]{\ensuremath{\cVar{Refd}{}_{#1}}\xspace}
\newcommand{\TCommitment}{\ensuremath{\contrFmt{\it TC}}\xspace}

\begin{document}

\title{Verifying liquidity of recursive Bitcoin contracts}

\author[M.~Bartoletti]{Massimo Bartoletti\rsuper{a}}
\author[S.~Lande]{Stefano Lande\rsuper{a}}
\author[M.~Murgia]{Maurizio Murgia\rsuper{b}}
\author[R.~Zunino]{Roberto Zunino\rsuper{b}}
\address{University of Cagliari, Italy}
\email{\texttt{\{bart,lande\}@unica.it}}
\address{University of Trento, Italy}
\email{\texttt{\{maurizio.murgia,roberto.zunino\}@unitn.it}}

\maketitle

\begin{abstract}
  Smart contracts --- 
  computer protocols that regulate the exchange of crypto-assets
  in trustless environments ---
  have become popular with the spread of blockchain technologies.
  A landmark security property of smart contracts is \emph{liquidity}: 
  in a non-liquid contract, it may happen that some assets remain frozen,
  \ie not redeemable by anyone.
  The relevance of this issue is witnessed by recent liquidity attacks 
  to Ethereum, which have frozen hundreds of USD millions.
  We address the problem of verifying liquidity on BitML, 
  a DSL for smart contracts with a secure compiler to Bitcoin,
  featuring primitives for currency transfers, 
  contract renegotiation and consensual recursion.
  Our main result is a verification technique for liquidity.
  We first transform the infinite-state semantics of BitML 
  into a finite-state one, 
  which focusses on the behaviour of a chosen set of contracts,
  abstracting from the moves of the context.
  With respect to the chosen contracts, this abstraction is sound,
  \ie if the abstracted contract is liquid, 
  then also the concrete one is such.
  We then verify liquidity by model-checking
  the finite-state abstraction. 
  We implement a toolchain that automatically verifies liquidity of 
  BitML contracts and compiles them to Bitcoin, 
  and we assess it through a benchmark of representative contracts.
\end{abstract}

\section{Introduction}
\label{sec:intro}

Smart contracts --- 
computer protocols that regulate the exchange of assets
in trustless environments ---
have become popular with the growth of interest in blockchain technologies.
Mainstream blockchain platforms like Ethereum, Tezos and Cardano, 
feature expressive high-level languages for programming smart contracts.
This flexibility has a drawback in that
it may open the door to attacks that steal or tamper with the
assets controlled by vulnerable contracts~\cite{ABC17post,Luu16ccs}.

An alternative approach is to sacrifice the expressiveness of smart contracts 
to reduce the attack surface.
This approach was pursued first by Bitcoin, where transactions 
can specify simple conditions on how to redeem them,
using a limited set of logic, arithmetic, and cryptographic operators.
Despite the limited expressiveness of these conditions, 
it is possible to encode a variety of smart contracts,
\eg gambling games, escrow services, crowdfunding systems,
by suitably chaining transactions 
\cite{Andrychowicz14bw,Andrychowicz14sp,Andrychowicz16cacm,bitcoinsok,Banasik16esorics,BZ17bw,Bentov14crypto,Kumaresan14ccs,KumaresanB16ccs,Kumaresan15ccs,KumaresanVV16ccs,Miller16zerocollateral}.
The common trait of these works is that they render contracts
as cryptographic protocols, where participants
can exchange/sign messages, read the blockchain, and append transactions.
Verifying the correctness of these protocols is hard, 
since it requires to reason in a computational model, 
where participants can manipulate arbitrary bitstrings,
only being constrained to use PPTIME algorithms.

Departing from this approach, BitML~\cite{BZ18bitml} 
allows to write Bitcoin contracts in a high-level, process-algebraic language.
BitML features a compiler that translates contracts 
into sets of standard Bitcoin transactions.
The compiler enjoys a computational soundness property, which guarantees that 
the execution of the compiled contract is coherent with 
the semantics of the source BitML specification,
even in the presence of adversaries.

In this paper we address the problem of verifying BitML contracts,
in an extension of BitML with renegotiation and recursion~\cite{BMZ20coordination}.
In particular, we focus on a landmark property of smart contracts,
called \emph{liquidity}, which ensures that funds cannot remain frozen
within a contract.
Before discussing our main contributions, we overview below
BitML and our analysis technique.

\subsection{BitML overview}

In BitML, any participant can broadcast a \emph{contract advertisement} 
$\contrAdv{\contrG}{\contrC}$,
where $\contrC$ is the contract,
specifiying the rules to transfer bitcoins (\BTC),
while $\contrG$ is a set of \emph{preconditions} to its stipulation.
Preconditions may require participants to deposit some $\BTC$ in the contract,
or to commit to some secret.
Once $\contrAdv{\contrG}{\contrC}$ has been advertised,
each participant can choose whether to accept it, or not.
When all the preconditions $\contrG$ have been satisfied, 
and all the involved participants have accepted,
the contract $\contrC$ becomes \emph{stipulated}.
Stipulated contracts have a \emph{balance}, 
initially set to the sum of the deposits required by its preconditions.
This balance is updated when participants execute the contract, 
\eg by depositing/withdrawing funds to/from the contract. 

A contract $\contrC$ is a \emph{choice} among zero or more branches.
Each branch is a \emph{guarded contract}, consisting of one
action, and zero or more continuations.
The guarded contract 
\[
\withdrawC{A}
\]
transfers the whole balance to $\pmvA$, and then terminates.
The guarded contract
\[
\splitC{\!\mmid_{i=1}^n \splitB{\valW[i]}{\contrC[i]}}
\]
decomposes the contract into $n$ parallel components $\contrC[i]$, 
distributing the balance according to the weights $\valW[i]$. 
The guarded contract 
\[
\putC[\predP]{}{\vec{\secrA}}
\]
checks that all the secrets $\vec{\secrA}$ have been revealed
and satisfy the predicate $\predP$
(of course, a secret can be revealed only by the participant
who has chosen it, as we will see in Section~\ref{sec:bitml}).

When enabled, the above-mentioned actions can be fired by anyone, at anytime. 
To restrict \emph{who} can execute actions and \emph{when}, 
one can use the decoration $\authC{\pmvA}{\contrD}$,
which requires the authorization of $\pmvA$,
and the decoration $\afterC{\constT}{\contrD}$,
which requires to wait until time $\constT$.

Finally, the guarded contract 
\[
\adv{\callX{\vec{\sexp}}}
\]
allows participants to renegotiate the contract. 
This requires first a definition for the variable $\callX{}$,
which is given by an equation of the form
$\decl{\cVarX}{\vec{\procParamA}} = \contrAdv{\contrG}{\contrC}$.
To finalize the renegotiation, 
all the participants involved in the current contract
must accept the new contract $\contrC$,%
\footnote{We use the asterisk in $\adv{\callX{\vec{\sexp}}}$ to stress
that \emph{all} the participants must accept the renegotiation.}
by satisfying its precondition $\contrG$ (similarly to stipulation).
When this happens, the control passes to $\contrC$, 
where the formal parameters $\vec{\procParamA}$ are instantiated
to the actual parameters $\vec{\sexp}$.
Note that $\contrC$ could refer to $\cVarX$, so enabling recursion.

We exemplify BitML by specifying the \emph{timed commitment} contract,
a basic protocol to construct more complex contracts, 
like \eg lotteries and other games~\cite{Andrychowicz14sp}.
Assume that a participant $\pmvA$ wants to choose a secret,
and promises to reveal it before some time $\constT$.
The contract must ensure that if $\pmvA$ does not reveal 
the secret in time, then she will pay a penalty of $1\BTC$ to $\pmvB$ 
(\eg, the opponent in a game).
In BitML, this is modelled as follows:
\[
\contrAdv
{\persdep{\pmvA}{1}{\varX} \mid \presecret{\pmvA}{\secrA}}
{\;
  ( 
  \putC{}{\secrA} . \, \withdrawC{\pmvA}
  \;\; + \;\;
  \afterC{\constT}{\withdrawC{\pmvB}}
  )
}
\]

The precondition requires $\pmvA$ to pay upfront $1\BTC$, 
deposited in a transaction $\varX$,
and to commit to a secret $\secrA$.
The contract is a guarded choice between two branches:
any participant can fire an enabled guard, and make the contract evolve to its continuation.
The guard of the left branch is $\putC{}{\secrA}$, which is enabled only 
after $\pmvA$ reveals the secret.
Its continuation $\withdrawC{\pmvA}$ allows anyone to transfer $1 \BTC$ to $\pmvA$.
The guard of the right branch is 
$\afterC{\constT}{\withdrawC{\pmvB}}$, 
which is enabled only after time $\constT$ and
allows anyone to transfer $1 \BTC$ to $\pmvB$
(here there is no continuation, since the $\withdrawC{}$ terminates).
So, before time $\constT$, 
$\pmvA$ has the option to reveal $\secrA$ (avoiding the penalty),
or to keep it secret (paying the penalty).
If no branch is taken by time $\constT$,
the first participant who fires its $\withdrawC{}$ gets $1\BTC$.

\subsection{Liquidity}

The \emph{liquidity} property requires that the contract balance 
can always be transferred eventually to some participant:
in a non-liquid contract, funds can be frozen forever, 
unavailable to anyone, hence effectively destroyed%
\footnote{To the best of our knowledge, the use of the term ``liquidity''
to refer to a contract property was first introduced in~\cite{Tsankov18ccs}, 
in the setting of Ethereum contracts.}.
A simple form of liquidity could just require that
participants can always cooperate to unfreeze funds.
However, this notion would contrast with the setting of smart contracts,
where participants are mutually untrusted, and may refuse to cooperate.
For instance, consider a contract
where $\pmvA$ and $\pmvB$ contribute $1\BTC$ each for a donation of $2\BTC$ 
to either $\pmvC$ or $\pmv{D}$:
\[
\contrAdv
{\persdep{\pmvA}{1}{\varX} \mid \persdep{\pmvB}{1}{\varY}}
{\;
\big(
\authC{\pmvA}{\authC{\pmvB}{\withdrawC{\pmvC}}} \; + \; \authC{\pmvA}{\authC{\pmvB}{\withdrawC{\pmv{D}}}}
\big)
\;}
\]
As in the timed commitment example, this contract is a choice between two branches, both decorated with $\authC{\pmvA}{\authC{\pmvB}{\cdots}}$. 
This means that taking any branch requires the authorization of both users: 
if $\pmvA$ and $\pmvB$ disagree on the branch to take, the funds are frozen.
When $\pmvA$ and $\pmvB$ agree on the recipient of the donation, the funds in the contract
are unlocked, and they can be transferred to the chosen recipient.

This contract would be liquid only by assuming the cooperation between $\pmvA$ and $\pmvB$:
indeed, $\pmvA$ alone cannot guarantee that the $2\BTC$ will eventually be donated,
as $\pmvB$ can choose a different recipient, or even refuse to give any authorization.
Consequently, unless $\pmvA$ trusts $\pmvB$, it makes sense to consider this contract as \emph{non-liquid} from $\pmvA$'s point of view 
(and for similar reasons, also from that of $\pmvB$).

Consider again the timed commitment contract
(we omit the preconditions for brevity):
\[
  \putC{}{\secrA} . \, \withdrawC{\pmvA}
  \;\; + \;\;
  \afterC{\constT}{\withdrawC{\pmvB}} 
\]

This contract is liquid from $\pmvA$'s point of view,
\emph{even} if $\pmvB$ is dishonest: 
indeed, $\pmvA$ can reveal the secret and then redeem the funds from the contract.
The timed commitment is also liquid from $\pmvB$'s point of view:
if $\pmvA$ does not reveal the secret (making the first branch stuck),
the funds in the contract can be redeemed through the second branch, after time $\constT$.

In a \emph{mutual} timed commitment,
where $\pmvA$ and $\pmvB$ have to exchange their secrets or pay a $1\BTC$ penalty,
achieving liquidity is a bit more challenging.
We first consider a wrong attempt:
\begin{align*}
  & \putC{}{\secrA} . \, \putC{}{\secrB} . \,  
    \splitC{(
    \splitB{1\BTC}{\withdrawC{\pmvA}}
    \mid
    \splitB{1\BTC}{\withdrawC{\pmvB}}
    )}
  \\
  & + \; \afterC{\constT}{\withdrawC{\pmvB}}
\end{align*}

This contract is liquid for $\pmvB$, but not for $\pmvA$.
Indeed, if $\pmvA$ performs the $\putC{}{\secrA}$ action, 
$\pmvB$ could refuse to reveal $\secrB$, making the contract stuck.
Instead, $\pmvB$ can wait time $\constT$ and then fire $\withdrawC{\pmvB}$;
if, in the meanwhile, $\pmvA$ has fired $\putC{}{\secrA}$,
$\pmvB$ can reveal his secret and fire $\putC{}{\secrB}$, 
and then liquidate the contract.

To make the contract liquid for both participants, we amend it as follows:
\begin{align*}
  &
    \putC{}{\secrA} . \, 
    \big(\putC{}{\secrB} . \, 
    \splitC{(
    \splitB{1\BTC}{\withdrawC{\pmvA}}
    \mid
    \splitB{1\BTC}{\withdrawC{\pmvB}}
    )}
  \\
  & \hspace {50pt} + \,
    \afterC{\constT+1}{\withdrawC{\pmvA}} 
    \big)
  \\
  & + \; \afterC{\constT}{\withdrawC{\pmvB}}
\end{align*}
Now, if $\pmvA$ has fired $\putC{}{\secrA}$ but $\pmvB$ refuses to reveal
$\secrB$, after time $\constT+1$ she can liquidate the contract
by performing the $\withdrawC{\pmvA}$.

As a more involved example, consider a recursive variant 
of the timed commitment:
\begin{align*}
  \decl{\cVarX}{n} 
  & = \contrAdv{\presecret{\pmvA}{\secrA} \mid \persdep{\pmvA}{1}{\varphX}}{\contrC}
  \\
  \contrC 
  & = \adv{\call{\cVarX}{n+1}}
    \; + \;
    \putC{}{\secrA} . \, \withdrawC{\pmvA}
    \; + \;
    \afterC{(\constT+n)}{\withdrawC{\pmvB}}
\end{align*}

This contract is a toy example of a recursive contract, where 
$\call{\cVarX}{n}$ can be renegotiated, 
transferring its balance to $\call{\cVarX}{n+1}$. 
When this happens, $\pmvA$ must commit to a new secret $\secrA$, 
and provide an additional deposit $\varphX$ of $1 \BTC$. 
Beyond renegotiation, $\call{\cVarX}{n}$ allows $\pmvA$ to reveal 
her secret and withdraw all the bitcoins deposited in the contract so far. 
If she does not reveal, $\pmvB$ can fire the last branch after time $\constT+n$,
transferring the whole balance to himself.
This contract is liquid for both $\pmvA$ and $\pmvB$.
In every reachable state, a participant can stop renegotiating the 
current contract $\call{\cVarX}{n}$.
Then, anyone can liquidate the contract
by either waiting until time $\constT + n$ and then performing 
the $\withdrawC{\pmvB}$, 
or firing the $\withdrawC{\pmvA}$ when this action is enabled.

The examples above, albeit elementary, 
show that detecting if a contract is liquid is not straightforward, in general.
Automatic techniques for the verification 
can be useful tools for the developers of smart contracts.

\subsection{Verifying liquidity}

One of the main contributions of this paper is
a verification technique for the liquidity of BitML contracts.
Our technique is based on a more general result,
\ie a correspondence between the concrete semantics of BitML
and a new abstract semantics, which is finite-state and
correctly approximates the concrete semantics.
To obtain a finite-state abstraction, 
we need to cope with several sources of infiniteness
of the concrete semantics:
the unbounded passing of time,
the stipulation and renegotiation of contracts,
and the operations on bitcoin deposits.
When studying the liquidity of a set of contracts $\VarX$ from the point of 
view of a participant $\pmvA$,
we abstract away all this, by just recording the actions which can be
performed on the descendants of $\VarX$, distinguishing between
the actions doable by $\pmvA$ alone from those which require 
the cooperation of other participants.
This abstraction produces a finite-state transition system, 
which we model-check for liquidity.

\subsection{Contributions}

We summarise our main contributions as follows:
\begin{itemize}

\item We introduce an extension of BitML featuring 
  the renegotiation primitive $\adv{\callX{}}$.
  Compared to the version in~\cite{BZ18bitml}, the current language
  is more expressive: 
  besides allowing participants to provide new deposits and secrets
  at run-time, it also allows for \emph{unbounded} recursion,
  still admitting compilation to Bitcoin.

\item We formalize a notion of liquidity (Definition~\ref{def:liquid}).
  With respect to a participant $\pmvA$, a contract is liquid 
  when $\pmvA$ alone can ensure that funds do not remain frozen within 
  the contract, even in the presence of adversaries.

\item We introduce an abstraction of the semantics of BitML 
  which is finite-state, and sound with respect to the concrete 
  (infinite-state) semantics.
  Building upon this abstraction,
  we devise a sound verification technique for liquidity in BitML
  (Theorem~\ref{th:abs-bitml:liquidity}).
  
\item We develop a toolchain for writing and verifying the liquidity of 
  BitML contracts, and for deploying them on Bitcoin.
  The toolchain is based on a BitML embedding in Racket~\cite{Flatt12cacm},
  which allows for programming BitML contracts within the DrRacket IDE.
  The toolchain also implements a compiler from BitML contracts to standard Bitcoin transactions.

\item We implement a collection of BitML contracts,
  which we use as a benchmark to evaluate our toolchain.
  This collection contains a variety of complex contracts, including
  financial services, auctions, lotteries, and other gambling games.

\item We discuss alternative renegotiation primitives,
  which allow participants to choose some parameters 
  (\eg the amounts to be deposited) at renegotiation time,
  to change the set of participants involved in the renegotiated contract,
  and to renegotiate contracts without the consent of all participants.

\item We discuss alternative notions of liquidity,
  \eg taking into account participants' strategies.

\end{itemize}

\subsection{Comparison with previous work}
This paper borrows and extends the contributions of some past papers of ours.
BitML was originally introduced in~\cite{BZ18bitml}, in a version without 
renegotiation and recursion.
A main limitation of this version was that the participants could not renegotiate the terms
of a stipulated contract: this prevented from expressing common financial contracts, where 
funds have to be added by participants at run-time.
Renegotiation and recursion were added in~\cite{BMZ20coordination},
where we showed that, despite the increased expressiveness,
it was still possible to execute BitML on standard Bitcoin,
preserving the security guarantees of BitML.
These papers did not deal with verification of contracts, and with liquidity.
This notion was introduced in \cite{BZ19post}, in the original version of BitML, 
\ie without renegotiation and recursion.
The BitML toolchain was first presented in~\cite{bitmlracket}, supporting the
compilation and verification of contracts in the original version of BitML.

The current paper is the first one which studies liquidity 
in the full BitML.
Recursion adds significant complexity to verification, as it makes the calculus Turing-complete.
Because of this, the abstraction in~\cite{BZ19post} is no longer usable,
so the current paper devises an alternative verification technique.
The current paper also improves the BitML toolchain, extending the compiler and the liquidity verifier to contracts with renegotiation and recursion. 

\section{BitML with renegotiation and recursion}
\label{sec:bitml}

We assume a set of \emph{participants}, 
ranged over by $\pmvA, \pmvB, \ldots$,
a set of \emph{deposit names} $\varX, \varY, \ldots$,
a set of \emph{deposit variables} $\varphX,\varphY,\hdots$,
and a set of \emph{secret names} $\secrA, \secrB, \ldots$.
We use $\dmv,\dmvi,\hdots$ to range over deposits (both names and variables),
and $\valV, \valVi, \valW$ to range over non-negative rational values.
We denote with $\PartT$ the set of the \emph{honest} participants.
We denote with $\secrOf{\pmvA}$ the set of secret names usable by $\pmvA$,
requiring that $\secrOf{\pmvA} \cap \secrOf{\pmvB} = \emptyset$
if $\pmvA \neq \pmvB$.

\begin{defi}[Contract precondition]
  \label{def:bitml:precondition}
  Contract preconditions have the following syntax
  (the deposits $\dmv$ in a contract precondition $\contrG$ must be distinct):
  \begin{align*}
    \contrG \bnfdef \;\;
      & \persdep{\pmvA}{\valV}{\dmv}
      && \text{deposit of $\valV \BTC$ put by $\pmvA$}
    \\
    \bnfmid\;
      & \presecret{\pmvA}{\secrA} 
      && \text{secret committed by $\pmvA$ ($\secrA \in \secrOf{\pmvA}$)}
    \\
    \bnfmid\;
      & \contrG \mid \contrG
      && \text{composition}
         \tag*{$\qedef$}
  \end{align*}
\end{defi}

The precondition $\persdep{\pmvA}{\valV}{\dmv}$ 
requires $\pmvA$ to own $\valV \BTC$ in a deposit $\dmv$,
and to spend it for stipulating the contract.
The precondition $\presecret{\pmvA}{\secrA}$ requires $\pmvA$ to generate a
secret $\secrA$, and commit to it before the contract starts. 
After stipulation, $\pmvA$ can choose whether to disclose 
the secret $\secrA$, or not.

To define contracts,
we assume a finite set of \emph{recursion variables}, 
ranged over by $\cVarX,\cVarY,\hdots$,
and a language of \emph{static expressions} $\sexp,\sexpi,\hdots$, 
formed by integer constants $k$, 
integer variables $\procParamA,\procParamB,\hdots$,
and the usual arithmetic operators.
We omit to define the syntax and semantics of static
expressions, since they are standard.
We assume that a closed static expression evaluates to an integer value.
We use the $\vec{\cdot}$ notation for finite sequences.

\begin{defi}[Contract]
  \label{def:bitml:contract}
  Contracts are terms with the syntax in~\Cref{fig:bitml:contract}, where:
  \begin{inlinelist}
  \item the summation $\textstyle \sum_{i \in I} \contrD[i]$ is over a finite set of indices $I$;
  \item each recursion variable $\cVarX$ has a unique defining equation 
    $\decl{\cVarX}{\vec{\procParamA}} = \contrAdv{\contrG}{\contrC}$;
  \item renegotiations $\adv{\callX{\vec{\sexp}}}$ 
    have the correct number of arguments;
  \item the names $\vec{\secrA}$ in
    $\putC[p]{}{\vec{\secrA}}$ are distinct, 
    and they include those in $p$.
  \end{inlinelist}
  We denote with $\cnil$ the empty sum. 
  The order of decorations is immaterial,
  \eg, $\afterC{\sexp}{\authC{A}{\authC{B}{\contrD}}}$ is equivalent to
  $\authC{B}{\authC{A}{\afterC{\sexp}{\contrD}}}$.
  \hfill{$\qedef$}
\end{defi}

A contract $\contrC$ is a choice among guarded contracts. 
The guarded contract $\withdrawC{A}$ transfers the whole balance to $\pmvA$.
A guarded contract $\putC[\predP]{}{\vec{\secrA}} .\, \contrCi$ 
continues as $\contrCi$
once all the secrets $\vec{\secrA}$ have been revealed
and satisfy the predicate $\predP$.
The guarded contract 
\mbox{$\splitC{(\splitB{\valW[1]}{\contrC[1]} \mid \cdots \mid \splitB{\valW[n]}{\contrC[n]})}$}
divides the contract into $n$ contracts $\contrC[i]$, 
splitting the balance according to the weights $\valW[i]$. 
To restrict \emph{who} can execute a branch and \emph{when}, 
one can use the decoration $\authC{\pmvA}{\contrD}$, 
requiring to wait for $\pmvA$'s authorization,
and the decoration $\afterC{\sexp}{\contrD}$,
requiring to wait until the time specified by the static
expression~$\sexp$.
The guarded contract $\adv{\callX{\vec{\sexp}}}$ allows the participants 
involved in the contract to renegotiate it.
Intuitively, if
$\decl{\cVarX}{\vec{\procParamA}} = \contrAdv{\contrG}{\contrC}$,
then the contract continues as 
$\contrC\setenum{\bind{\vec{\procParamA}}{\vec{\sexp}}}$  
if all the participants 
mentioned in $\contrG$ give their authorization,
and satisfy the precondition $\contrG$.

\begin{figure*}[t]
  \footnotesize
  \begin{minipage}{0.62\textwidth}
    \begin{align*}
      \contrC 
                 & \bnfdef 
                   \textstyle \sum_{i \in I} \contrD[i]
                                                       && \text{contract}
      \\[4pt]
      \contrD
                 & \bnfdef
                                                       && \text{guarded contract}
      \\
                 & \withdrawC{\pmvA}
                                                       && \text{transfer the balance to $\pmvA$}
      \\[2pt]
      \bnfmid
                 & \putC[\predP]{}{\vec{\secrA}} . \, \contrC
                                                       && \text{reveal secrets (if $\predP$ is true)} 
      \\[2pt]
      \bnfmid
                 & \splitC{\|_i\, (\splitB{\valW[i]}{\contrC[i]})}
                                                       && \text{split the balance}
      \\[2pt]
      \bnfmid
                 & \authC{\pmvA}{\contrD}
                                                       && \text{wait for $\pmvA$'s authorization}
      \\[2pt]
      \bnfmid
                 & \afterC{\sexp}{\contrD}
                                                       && \text{wait until time $\sexp$}
      \\[2pt]
      \bnfmid
                 & \adv{\callX{\vec{\sexp}}}
                                                       && \text{renegotiate the contract}
    \end{align*}
  \end{minipage}
  \begin{minipage}{0.35\textwidth}
    \begin{align*}
      \predP 
      \bnfdef \; &
           \true
      && \text{truth}
      \\
      \bnfmid
                 & \predP \land \predP
        && \text{conjunction}
      \\
      \bnfmid
                 & \neg \predP
        && \text{negation}
      \\
      \bnfmid
                 & \expE = \expE 
        && \text{equality} 
      \\
      \bnfmid
                 & \expE < \expE 
        && \text{less than}
      \\
      \expE
      \bnfdef \; &
           \sexp
      && \text{static expression}
      \\
      \bnfmid
                 & \secrA
        && \text{secret}
      \\
      \bnfmid
                 & \expE + \expE
        && \text{addition} 
      \\
      \bnfmid
                 & \expE - \expE
        && \text{subtraction}
    \end{align*}
  \end{minipage}
  \caption{Syntax of BitML contracts.}
  \label{fig:bitml:contract}
\end{figure*}

\begin{defi}[Contract advertisement]
  \label{def:bitml:contrAdv}
  A contract advertisement is a term $\contrAdv[\circ]{\contrG}{\contrC}$,
  such that:
  \begin{inlinelist}
  \item $\circ$ is either empty or a deposit name;
  \item each secret name in $\contrC$ occurs in $\contrG$; 
  \item \label{item:bitml:contrAdv:persistent}
    $\contrG$ requires a deposit from each 
    $\pmvA$ in $\contrAdv[\circ]{\contrG}{\contrC}$;
  \item each $\adv{\callX{\vec{\sexp}}}$ in $\contrC$ refers to
    a defining equation
    $\decl{\cVarX}{\vec{\procParamA}} = \contrAdv{\contrGi}{\contrCi}$
    where the participants in $\contrGi$ are the same as those in
    $\contrG$.
    \hfill{$\qedef$}
  \end{inlinelist}
\end{defi}

Intuitively, $\contrAdv{\contrG}{\contrC}$ is the advertisement of 
a contract $\contrC$ with preconditions $\contrG$,
while $\contrAdv[\varX]{\contrG}{\contrC}$ is the advertisement of 
a renegotiation of an existing contract $\varX$.
Condition~\ref{item:bitml:contrAdv:persistent} 
is used to guarantee that the contract is stipulated 
only if \emph{all} the involved participants give their consent:
namely, $\pmvA$'s consent is rendered as $\pmvA$'s authorization to spend one 
of her deposits.
The last condition is only used to simplify the technical development:
we outline in Section~\ref{sec:variants-bitml} how to relax it, 
by allowing renegotiations to exclude some participants,
or to include new ones, which were not among those who originally 
stipulated the contract.

We now define the semantics of BitML, starting from its configurations.

\begin{defi}[Configuration]
  \label{def:bitml:conf}
  Configurations are terms with the syntax in~\Cref{fig:bitml:conf}, where:
  \begin{inlinelist}
  \item in a committed secret, $N \in \Nat \cup \setenum{\bot}$
    (where $\bot$ denotes an ill-formed commitment);
  \item in a revealed secret, $N \in \Nat$;
  \item \label{def:bitml:conf:no-dup-auth}
    in a configuration there are no duplicate authorizations;
  \item in a configuration containing $\confContr[\varX]{\cdots}{}$ and
    $\confContr[\varY]{\cdots}{}$, it must be $\varX \neq \varY$;
  \item there exists at most one term $\constT$.
  \end{inlinelist}
  We assume that $(\mid,\gnil)$ is a commutative monoid, and 
  we denote indexed parallel compositions with~$\mmid_i$.
  We say that $\confG$ is \emph{initial} when it contains only deposits 
  (\ie, terms $\confDep[x]{\pmvA}{\valV}$),
  and that it is a \emph{timed configuration} 
  when it contains a term $\constT$.
  We denote with $\cn{\confG}$ the set of contract names $\varX$
  such that $\confG$ contains $\confContr[\varX]{\contrC}{\valV}$,
  for some $\contrC$ and $\valV$.
  \hfill\qedef
\end{defi}

The intuition behind the various terms in configurations is the following:
\begin{itemize}

\item $\confContr[\varX]{\contrC}{\valV}$
  is a stipulated contract storing $\valV \BTC$, 
  uniquely identified by the name $\varX$;

\item $\confDep[\varX]{\pmvA}{\valV}$ is a deposit of $\valV \BTC$ 
  owned by $\pmvA$, and uniquely identified by the name $\varX$;

\item $\confAuth{\pmvA}{\authmv}$
  is $\pmvA$'s \emph{authorizations} to perform some action $\authmv$;

\item \mbox{$\confSec{\pmvA}{\secrA}{N}$}
  represents $\pmvA$'s commitment to a secret $N$, identified by $\secrA$;

\item $\confRev{\pmvA}{\secrA}{N}$
  represents a secret $N$, identified by $\secrA$, and revealed by $\pmvA$.

\end{itemize}

\begin{figure*}
  \footnotesize
  \resizebox{\textwidth}{!}{
    \begin{minipage}{0.45\textwidth}
      \begin{align*}
        \confG
        & \bnfdef
        && \text{configuration}
        \\
        & \cnil
        && \text{empty} 
        \\
        \bnfmid
        & \contrAdv[\circ]{\contrG}{\contrC}
        && \text{contract advertisement}
        \\
        \bnfmid
        & \confContr[x]{\contrC}{\valV}
        && \text{active contract storing $\valV \BTC$}
        \\
        \bnfmid
        & \confDep[x]{\pmvA}{\valV}
        && \text{$\pmvA$'s deposit of $\valV \BTC$}
        \\
        \bnfmid
        & \confAuth{\pmvA}{\authmv}
        && \text{$\pmvA$'s authorization for $\authmv$}
        \\
        \bnfmid
        & \confSec{\pmvA}{\secrA}{N}
        && \text{committed secret of $\pmvA$}
        \\
        \bnfmid
        & \confRev{\pmvA}{\secrA}{N}
        && \text{revealed secret of $\pmvA$}
        \\
        \bnfmid
        & \timeT
        && \text{global time}
        \\
        \bnfmid
        & \confG \mid \confGi
        && \text{composition}
      \end{align*}
    \end{minipage}
    \begin{minipage}{0.4\textwidth}
      \begin{align*}
        \authmv 
        & \bnfdef
        && \text{authorization to \ldots}
        \\
        & \authCommit[\circ]{\contrG}{\contrC}
        && \text{commit to $\contrAdv[\circ]{\contrG}{\contrC}$}
        \\
        \bnfmid
        & \authAdv[\circ]{\varX}{\contrG}{\contrC}
        && \text{spend $\varX$ for $\contrAdv[\circ]{\contrG}{\contrC}$}
        \\
        \bnfmid
        & \authBranch{\varX}{\contrD}
        && \text{take branch $\contrD$ in contract $\varX$}
        \\
        \bnfmid
        & \authJoin{\varX}{\varY}{\pmvA}{\valV}
        && \text{join deposits $x$ and $y$}
        \\      
        \bnfmid
        & \authSplit{x}{\pmvA}{\valV}{\valVi}
        && \text{split deposit $x$ in two}
        \\
        \bnfmid
        & \authDonate{x}{\pmvB}
        && \text{donate deposit $x$ to $\pmvB$}
        \\
        \bnfmid
        & \authDestroy{\vec{\varX}}{i}{\varY}
        && \text{destroy $i$-th deposit in $\vec{x}$}
      \end{align*}
    \end{minipage}}
  \caption{Configurations and authorizations.}
  \label{fig:bitml:conf}
\end{figure*}

\begin{defi}[BitML semantics]
  \label{def:bitml:timed-LTS}
  The semantics of BitML is a Labelled Transition System (LTS)
  between timed configurations.
  In the rest of this~\namecref{sec:bitml} we describe 
  the reduction rules of the LTS, 
  which implicitly define the labels $\labS$.
  A \emph{concrete run} $\runS$ is a sequence
  \(
  \confG[0]
  \concrightarrow{\labS[0]}
  \confG[1]
  \concrightarrow{\labS[1]}
  \cdots
  \),
  where $\confG[0]$ is timed.
  If $\runS$ is finite, we write 
  $\confS{\runS}$ for the untimed part of its last configuration.
\end{defi}

Below we gently introduce the BitML semantics, 
first illustrating each construct through examples, 
and then giving the general rule.
Labels represent the actions performed by participants. 
A decoration $\authLab{\pmvA}{\cdots}$  
in the label means that the action can be performed only by $\pmvA$.
The absence of such a decoration means that the action can be performed by anyone.
Note that labels are not instrumental to define the BitML semantics: yet, they are essential to the definition of liquidity, since there we need to associate actions to the participants who can perform them.
In the examples, we will omit the labels.

\mypar{Deposits}

A deposit $\confDep[x]{\pmvA}{\valV}$ can be reduced in several ways:
it can be split into smaller deposits, joined with another deposit,
transferred to another participant, or destroyed. 
In all cases, its owner $\pmvA$ must first authorise the action.
The reduction rules for deposits are detailed in~\Cref{fig:bitml:semantics:dep}.
Rule \nrule{[Dep-AuthJoin]} allows $\pmvA$ 
to authorize the merge of two deposits $x,y$ into a single one, 
creating the needed authorization. 
The label of the form $\authLab{\pmvA}{\cdots}$ records that only $\pmvA$ 
can perform this move. 
Rule \nrule{[Dep-Join]} uses this authorization to create a single deposit $z$ of $\pmvA$. 
The rules \nrule{[Dep-AuthDivide]} and \nrule{[Dep-Divide]} act similarly, 
allowing a deposit of $\pmvA$ to be divided in two parts. 
The rules \nrule{[Dep-AuthDonate]} and \nrule{[Dep-Donate]} allow $\pmvA$ to transfer one of her deposits to another participant.
The rules~\nrule{[Dep-AuthDestroy]} and~\nrule{[Dep-Destroy]} allow a set of participants to destroy a set of deposits 
$x_1 \cdots x_n$. 
To do that, first each participant $\pmvA[i]$ must provide 
the needed authorization
$\confAuth{\pmvA[i]}{\authDestroy{\vec{x}}{i}{y}}$
for their own deposit $\varX[i]$.
When all the authorizations have been collected, rule~\nrule{[Dep-Destroy]}
eliminates the deposits.
The last two rules in~\Cref{fig:bitml:semantics:dep} 
are needed to properly represent the fact that computational participants
can create (and put on the ledger) transactions without a counterpart 
in the symbolic model.
To achieve a meaningful correspondence between the symbolic and the computational models, 
putting on the ledger such transactions is rendered  with the rule~\nrule{[Dep-destroy]}.

\begin{figure*}[h!]
  \small
  \begin{tabular}{c}
    \(
    \irule
    {}
    {\confDep[x]{\pmvA}{\valV} \mid \confDep[y]{\pmvA}{\valVi} \mid \confG
    \concrightarrow{\authLab{\pmvA}{x,y}} 
    \confDep[x]{\pmvA}{\valV} \mid \confDep[y]{\pmvA}{\valVi} \mid 
    \confAuth{\pmvA}{\authJoin{x}{y}{\pmvA}{\valV + \valVi}}  \mid \confG
    }
    \smallnrule{[Dep-AuthJoin]}
    \)
    \\[20pt]
    \(
    \irule
    {\confG = \confAuth{\pmvA}{\authJoin{x}{y}{\pmvA}{\valV + \valVi}} \mid \confAuth{\pmvA}{\authJoin{y}{x}{\pmvA}{\valV + \valVi}} 
    \mid \confGi
    \quad z \; \text{fresh}}
    {\confDep[x]{\pmvA}{\valV} \mid \confDep[y]{\pmvA}{\valVi} \mid \confG
    \concrightarrow{{\it join}(x,y)} 
    \confDep[z]{\pmvA}{\valV + \valVi} \mid \confGi
    }
    \smallnrule{[Dep-Join]}
    \)
    \\[20pt]
    \(
    \irule
    {}
    {\confDep[x]{\pmvA}{\valV + \valVi} \mid \confG
    \concrightarrow{\authLab{\pmvA}{\varX,\valV,\valVi}} 
    \confDep[x]{\pmvA}{\valV+\valVi} \mid 
    \confAuth{\pmvA}{\authSplit{\varX}{\pmvA}{\valV}{\valVi}} \mid \confG
    }
    \smallnrule{[Dep-AuthDivide]}
    \)
    \\[20pt]
    \(
    \irule
    {\confG = \confAuth{\pmvA}{\authSplit{x}{\pmvA}{\valV}{\valVi}} 
    \mid \confGi
    \quad y,y' \; \text{fresh}}
    {\confDep[x]{\pmvA}{\valV + \valVi} \mid \confG
    \concrightarrow{{\it divide}(\varX,\valV,\valVi)} 
    \confDep[y]{\pmvA}{\valV} \mid \confDep[y']{\pmvA}{\valVi}
    \mid \confGi
    }
    \smallnrule{[Dep-Divide]}
    \)
    \\[20pt]
    \(
    \irule
    {}
    {\confDep[\varX]{\pmvA}{\valV} \mid \confG
    \concrightarrow{\authLab{\pmvA}{\varX,\pmvB}} 
    \confDep[\varX]{\pmvA}{\valV} \mid 
    \confAuth{\pmvA}{\authDonate{\varX}{\pmvB}} \mid \confG
    }
    \smallnrule{[Dep-AuthDonate]}
    \)
    \\[20pt]
    \(
    \irule
    {\confG = \confAuth{\pmvA}{\authDonate{x}{\pmvB}} 
    \mid \confGi
    \quad \varY \; \text{fresh}}
    {\confDep[\varX]{\pmvA}{\valV} \mid \confG
    \concrightarrow{{\it donate}(\varX,\pmvB)} 
    \confDep[\varY]{\pmvB}{\valV}
    \mid \confGi
    }
    \smallnrule{[Dep-Donate]}
    \)
    \\[20pt]
    \(
    \irule
    {\vec{x} = x_1 \cdots x_n \quad j \in 1..n \quad
    y \text{ fresh (except in destroy authorizations for $\vec{x}$)}
    }
    {\big( \mmid_{i=1}^{n} \confDep[x_i]{\pmvA[i]}{\valV[i]} \big) \mid \confG
    \concrightarrow{\authLab{\pmvA[j]}{\vec{\varX},j}}
    \big( \mmid_{i=1}^{n} \confDep[x_i]{\pmvA[i]}{\valV[i]} \big) \mid 
    \confAuth{\pmvA[j]}{\authDestroy{\vec{x}}{j}{y}} \mid 
    \confG
    }
    \smallnrule{[Dep-AuthDestroy]}
    \)
    \\[20pt]
    \(
    \irule
    {\vec{x} = x_1 \cdots x_n \quad 
    \confG = \big( \mmid_{i=1}^{n} \, \confAuth{\pmvA[i]}{\authDestroy{\vec{x}}{i}{y}} \big) \mid \confGi}
    {\big( \mmid_{i=1}^{n} \confDep[x_i]{\pmvA[i]}{\valV[i]} \big) \mid \confG
    \concrightarrow{{\it destroy}(\vec{\varX})}
    \confGi}
    \smallnrule{[Dep-Destroy]}
    \)
  \end{tabular}
  \caption{Semantics of Bitcoin deposits.}
  \label{fig:bitml:semantics:dep}
\end{figure*}

\mypar{Stipulation: advertisement.}

Any participant can broadcast a contract advertisement
$\contrAdv{\contrG}{\contrC}$,
provided that all the deposits mentioned in $\contrG$ exist 
in the current configuration, 
and that the names of the secrets in $\contrG$ are fresh.
This is formalised by the following rule:
\medskip
\[
\irule
{
  \forall \persdep{\pmvA}{\valV}{\varX} \text{ in } \contrG :
  \confDep[{\varX}]{\pmvA}{\valV} \text{ in } \confG
  \qquad
  \text{ all secrets in } \contrG \text{ fresh }
}
{\confG 
  \concrightarrow{{\it adv}(\contrAdv{\contrG}{\contrC})}
  \contrAdv{\contrG}{\contrC} \mid \confG}
\smallnrule{[C-Adv]}
\]

\noindent
We exemplify this and the following rules through a running example.
Let:
\[
\contrG = 
\persdep{\pmvA}{1}{\varX} \mid 
\persdep{\pmvB}{1}{\varphX} \mid 
\presecret{\pmvA}{\secrA}
\]
This precondition requires both $\pmvA$ and $\pmvB$ to deposit
$1\BTC$, but $\pmvA$'s deposit name $\varX$ is known,
while $\pmvB$'one is not known yet, so we refer to it  
through a deposit variable $\varphX$.
Let $\contrC$ be an arbitrary contract involving only $\pmvA$ and $\pmvB$, and 
let $\confG[0] = \confG \mid \confDep[\varX]{\pmvA}{1} \mid \confDep[\varY]{\pmvB}{1}$
for some~$\confG$.
By rule \nrule{[C-Adv]}, the configuration $\confG[0]$ can take the transition:
\[
\confG[0] \;\concrightarrow{}\; \confG[0] \mid \contrAdv{\contrG}{\contrC} 
\; = \; \confG[1]
\]

\mypar{Stipulation: commitment}

To stipulate an advertised contract, all the participants mentioned in it
must fulfill the preconditions, by making available the required deposits, 
and committing to the required secrets.
In our example, $\pmvA$ has one secret to commit,  
so she can perform the following step:
\begin{equation}
  \label{eq:bitml:AuthCommitA}
  \confG[1]
  \; \concrightarrow{} \;
  \confG[1]
  \mid \confSec{\pmvA}{\secrA}{N}
  \mid \confAuth{\pmvA}{\authCommit{\contrG}{\contrC}}
  \; = \; \confG[2]
\end{equation}
where the term $\confSec{\pmvA}{\secrA}{N}$ represents  
$\pmvA$'s commitment to the secret $N$, while 
$\confAuth{\pmvA}{\authCommit{\contrG}{\contrC}}$ 
represents finalising the commitment phase for $\pmvA$.
Participant $\pmvB$ has no secrets to commit, but he must choose
one of his deposits (\eg, $\confDep[\varY]{\pmvB}{1}$)
to fulfill the precondition $\persdep{\pmvB}{1}{\varphX}$:
\begin{equation}
  \label{eq:bitml:AuthCommitB}
  \confG[2]
  \; \concrightarrow{} \;
  \confG[2]
  \mid \assign{\pmvB}{\varphX}{\varY}
  \mid \confAuth{\pmvB}{\authCommit{\contrG}{\contrC}}
  \; = \; \confG[3]
\end{equation}
In general, these steps are formalised by the following rule:
\[
\small
\irule
{\begin{array}{c}
   \begin{array}{l}
     \confD[s] = \mmid_{i=1}^{k} \; \confSec{\pmvA}{\secrA[i]}{N_i} 
     \\[3pt]
     \secrA[1] \cdots \secrA[k] \text{ are all the secrets of $\pmvA$ in $\contrG$} 
     \\[3pt]
     \forall i \in 1..k : \nexists N : 
     \confSec{\pmvA}{\secrA[i]}{N} \text{ in } \confG
     \\[3pt]
     \forall i \in 1..k : \nexists N : 
     \confRev{\pmvA}{\secrA[i]}{N} \text{ in } \confG
     \\[3pt]
     \forall i \in 1..k : N_i \in \begin{cases}
       \Nat & \text{if $\pmvA \in \PartT$} \\
       \Nat \cup \setenum{\bot} & \text{otherwise}
     \end{cases}
   \end{array}
   \begin{array}{l}
     \confD[d] = \mmid_{i=1}^{h} \;
     \assign{\pmvA}{\varphX[i]}{\varX[i]} 
     \\[3pt]
     \varphX[1] \cdots \varphX[h] \text{ are all the deposits of $\pmvA$ in $\contrG$}
     \\[3pt]
     \forall i \in 1..h : \persdep{\pmvA}{\valV[i]}{\varphX[i]}
     \text{ in } \contrG
     \\[3pt]
     \forall i \in 1..h : \nexists \varX : (\assign{\pmvA}{\varphX[i]}{\varX}) 
     \text{ in } \confG 
     \\[3pt]
     \forall i \in 1..h : \confDep[{\varX[i]}]{\pmvA}{\valV[i]} \text{ in } \confG
     \\[3pt]
     \forall i \neq j \in 1..h : \varX[i] \neq \varX[j]
     \\[3pt]
     \forall \varX, \valV, i \in 1..h : 
     \persdep{\pmvA}{\valV}{\varX} \text{ in } \contrG
     \implies {\varX[i]} \neq \varX
   \end{array}
 \end{array}
}
{\contrAdv[\circ]{\contrG}{\contrC} \mid \confG
  \concrightarrow{\authLab{\pmvA}{\contrAdv[\circ]{\contrG}{\contrC},\confD[s] \mid \confD[d]}}
  \contrAdv[\circ]{\contrG}{\contrC} 
  \mid \confG \mid \confD[s] \mid \confD[d]
  \mid
  \confAuth{\pmvA}{\authCommit[\circ]{\contrG}{\contrC}}
}
\smallnrule{[C-AuthCommit]}
\]

\medskip
The rule preconditions ensures that the final configuration fulfills the 
conditions required by $\contrG$.
We use the notation $\contrAdv[\circ]{\contrG}{\contrC}$ to refer to a contract 
advertisement of the form $\contrAdv{\contrG}{\contrC}$ or
$\contrAdv[\varX]{\contrG}{\contrC}$.
The case $\contrAdv{\contrG}{\contrC}$ corresponds to the original 
contract stipulation, while the other case corresponds to renegotiation.
In this way, we can use the same rule in both situations.
Note that condition~\ref{def:bitml:conf:no-dup-auth}
in Definition~\ref{def:bitml:conf} ensures that rule \nrule{[C-AuthCommit]}
cannot be used more than once to generate the same authorization.
The same is true for all the other rules that generate authorizations.

\mypar{Stipulation: authorization}
Back to our example, in the configuration $\confG[3]$ 
of~\eqref{eq:bitml:AuthCommitB},
$\pmvA$ and $\pmvB$ must perform an additional sequence of steps
to authorize the transfer of their deposits $\varX$, $\varY$ to the contract:
\begin{equation}
  \label{eq:bitml:AuthInitDep}
  \confG[3]
  \; \concrightarrow{} \;
  \confG[3]
  \mid \confAuth{\pmvA}{\authAdv{\varX}{\contrG}{\contrC}}
  \; \concrightarrow{} \;
  \confG[3]
  \mid \confAuth{\pmvA}{\authAdv{\varX}{\contrG}{\contrC}}
  \mid \confAuth{\pmvB}{\authAdv{\varY}{\contrG}{\contrC}}
  \; = \; \confG[4]
\end{equation}
where the terms $\confAuth{\pmvA}{\authAdv{x}{\contrG}{\contrC}}$ and 
$\confAuth{\pmvB}{\authAdv{y}{\contrG}{\contrC}}$ represent
the authorizations to spend $\varX$ and $\varY$ for stipulation.
In general, these steps are obtained through the following rule:
\[
\irule
{
  \forall \pmvB \text{ in } \contrG :
  \confAuth{\pmvB}{\authCommit[\circ]{\contrG}{\contrC}}
   \text{ in } \confG
  \qquad
  \text{$\persdep{\pmvA}{\valV}{\dmv} \text{ in } \contrG$}
  \qquad
  \confG \vdash \dmv = \varX}
{\contrAdv[\circ]{\contrG}{\contrC} \mid \confG
  \concrightarrow{\authLab{\pmvA}{\contrAdv[\circ]{\contrG}{\contrC},x}}
  \contrAdv[\circ]{\contrG}{\contrC} \mid \confG \mid \confAuth{\pmvA}{\authAdv[\circ]{x}{\contrG}{\contrC}}
}
\smallnrule{[C-AuthInitDep]}
\]

\noindent
The first premise requires that all participants have finalised the 
commitment phase.
While rule \nrule{[C-AuthCommit]} allows a participant to add all her
commitments to the configuration in a single step,
each application of rule~\nrule{[C-AuthInitDep]} allows one 
to authorize the spending of a single deposit.
If a deposit variable $\varphX$ was used in $\contrG$, 
the relation $\confG \vdash \varphX = \varX$ ensures that 
the configuration contains a binding $\assign{\pmvA}{\varphX}{\varX}$.
The relation $\vdash$ is defined as follows:
\[
\confG \;\vdash\; \varX = \varX \qquad 
\confG \mid \assign{\pmvA}{\varphX}{\varX} \;\vdash\; \varphX = \varX
\]

\mypar{Stipulation: initialization}

In the configuration $\confG[4]$ of~\eqref{eq:bitml:AuthInitDep} 
all the needed authorizations have been granted, 
so the advertisement can be turned into an active contract. 
This step consumes the deposits and the authorizations,
and it initializes the new contract, with a fresh name $\varZ$,
and with a balance corresponding to the sum of all the consumed deposits:
\[
\confG[4]
\; \concrightarrow{} \;
\confG \mid \confSec{\pmvA}{\secrA}{N} \mid \confContr[z]{\contrC}{2}
\]
In general, this step is obtained through the following rule:
\begin{equation*}
  \irule{
    \begin{array}{l}
      \contrG = 
      \big( \mmid_{i \in I} \persdep{\pmvA[i]}{\valV[i]}{\varX[i]} \big) \mid
      \big( \mmid_{i \in J} \persdep{\pmvB[i]}{\valVi[i]}{\varphX[i]} \big) \mid
      \big( \mmid_{i \in K} \presecret{\pmvC[i]}{\secrA[i]} \big)
      \qquad 
      \varZ \text{ fresh}
      \\[3pt]
      \confD = \big( \mmid_{i \in I} \confDep[{\varX[i]\,}]{\pmvA[i]}{\valV[i]} \big)
      \mid
      \big( \mmid_{i \in J} \confDep[{\varXi[i]\,}]{\pmvB[i]}{\valVi[i]} \big)
      \mid
      \big( \mmid_{i \in J} \assign{\pmvB[i]}{\varphX[i]}{\varXi[i]} \big)
      \mid
      \\[3pt]
      \hspace{24pt}
      \big( \mmid_{i \in I} \confAuth{\pmvA[i]}{\authAdv{\varX[i]}{\contrG}{\contrC}} \big)
      \mid
      \big( \mmid_{i \in J} \confAuth{\pmvB[i]}{\authAdv{\varXi[i]}{\contrG}{\contrC}} \big) \mid
      \big( \mmid_{\pmvA \in \contrG} \confAuth{\pmvA}{\authCommit{\contrG}{\contrC}} \big)
    \end{array}
  }
  {\contrAdv{\contrG}{\contrC} \mid \confD \mid \confG
    \concrightarrow{{\it init}(\contrG,\contrC)}
    \confContr[\varZ]{\contrC}{\sum_{i \in I} \valV[i] + \sum_{i \in J} \valVi[i]}
    \mid \confG
  }
  \smallnrule{[C-Init]}
\end{equation*}

\noindent
Note that the part \mbox{$\contrAdv{\contrG}{\contrC} \mid \confD$} 
of the configuration contains all the terms that are consumed by the step.
The following rules define the behaviour of a contract after stipulation.

\mypar{Withdraw}

Executing $\withdrawC{\pmvA}$ terminates the contract,
and transfers its balance to~$\pmvA$:
\[
\confContr[x]{\withdrawC{\pmvA}}{\valV}   
\;\concrightarrow{}\;
\confDep[y]{\pmvA}{\valV} 
\]
After the contract $\varX$ is terminated,
a fresh deposit of $\valV\BTC$ owned by $\pmvA$ is created.
The general rule is the following:
\[
\irule
{\varY \text{ fresh}}
{\confContr[\varX]{\withdrawC{\pmvA}}{\valV} \mid \confG
  \concrightarrow{{\it withdraw}(\pmvA,\valV,\varX)}
  \confContr[\varY]{\pmvA}{\valV}
  \mid \confG}
\smallnrule{[C-Withdraw]}
\]
The case where the action $\withdrawC{\pmvA}$ has an alternative branch 
is dealt with by the rule \nrule{[C-Branch]}, discussed below.

\mypar{Split}

The $\splitname$ primitive divides the contract balance in parts,
each one controlled by its own contract.
For instance:
\[
\confContr[x]
{(\splitC{\splitB{2}{\contrC[1]} \mid \splitB{3}{\contrC[2]}})}
{5} 
\;\concrightarrow{}\;
\confContr[y]{\contrC[1]}{2} \mid 
\confContr[z]{\contrC[2]}{3}
\]
After this step, the new spawned contracts $\contrC[1]$ and $\contrC[2]$
are executed concurrently.
The general rule is the following:
\[
\irule
{
  \valW = \sum_{i=1}^{k} \valW[i]
  \qquad
  \valV[i] = (\valV \cdot \valW[i])/\valW
  \qquad
  \varY[1] \cdots \varY[k] \text{ fresh}
}
{\confContr[\varX]{
    \splitC{\!\mmid_{i=1}^k (\splitB{\valW[i]}{\contrC[i]})}
  }{\valV}
  \mid \confG
  \concrightarrow{{\it split}(\varX)}
  \big( \mmid_{i=1}^{k} \confContr[{\varY[i]}]{\contrC[i]}{\valV[i]} \big) 
  \mid \confG
}
\smallnrule{[C-Split]}
\]

\noindent
Note that the weights $\valW[i]$ in the $\splitname$ do not represent
actual $\BTC$ values, but the proportion \wrt the contract balance.
For instance:
\[
\confContr[x]
{(\splitC{\splitB{2}{\contrC[1]} \mid \splitB{3}{\contrC[2]}})}
{10} 
\;\concrightarrow{}\;
\confContr[y]{\contrC[1]}{4} \mid 
\confContr[z]{\contrC[2]}{6}
\]

\mypar{Revealing secrets}

Any participant can reveal one of her secrets, using the rule:
\[
\small
\irule
{N \neq \bot}
{\confSec{\pmvA}{\secrA}{N} \mid \confG
  \concrightarrow{\authLab{\pmvA}{a}}
  \confRev{\pmvA}{\secrA}{N} \mid \confG}
\smallnrule{[C-AuthRev]}
\]

\noindent
The premise $N \neq \bot$ is needed to avoid the case where
a participant does not know the secret she has committed to.
Indeed, at the level of Bitcoin, commitments are represented as 
cryptographic hashes of bitstrings, 
and revealing a secret amounts to broadcasting a preimage, 
\ie a value whose hash is equal to the committed value.
If a participant commits to a random value, 
then with overwhelming probability she will not be able to 
provide a preimage.
The label $\authLab{\pmvA}{a}$ represents the fact that only 
$\pmvA$, the participant who performed the commitment, 
can fire the transition.

\mypar{Reveal}

The prefix $\putC[\predP]{}{\vec{\secrA}}$ can be fired if all the  
committed secrets $\vec{\secrA}$ have been revealed, 
and satisfy the guard $\predP$.
For instance, 
if $\confG = \confRev{\pmvA}{\secrA}{N} \mid \confRev{\pmvB}{\secrB}{M}$:
\[
\confContr[x]
{\putC[\secrA=\secrB]{}{\secrA\secrB}.\, \contrC}
{\valV}
\mid 
\confG
\; \concrightarrow{} \;
\confContr[y]{\contrC}{\valV}
\mid 
\confG
\tag*{if $M = N$}
\]

\noindent
The general rule is the following:
\[
\irule
{
  \vec{\secrA} = \secrA[1] \cdots \secrA[n]
  \qquad
  \confD = \mmid_{i=1}^{n} \confRev{\pmvB[i]}{\secrA[i]}{N_i}
  \qquad
  \sem[\confD]{\predP} = \true
  \qquad
  \varY \text{ fresh}
}
{\confContr[\varX]{\putC[\predP]{}{\vec{\secrA}}.\,\contrC}{\valV}
  \mid \confG \mid \confD
  \concrightarrow{{\it rev}(\vec{\secrA},\varY)}
  \confContr[\varY]{\contrC}{\valV}
  \mid \confG \mid \confD 
}
\smallnrule{[C-Rev]}
\]
where the semantics of predicates $\sem[\confD]{\predP}$
is defined by the following equations:
\[
\begin{array}{c}
  \sem[\confD]{\true} 
  = \true
  \qquad
  \sem[\confD]{\predP[1] \land \predP[2]} 
  = \sem[\confD]{\predP[1]} \text{ and } \sem[\confD]{\predP[2]}
  \qquad
  \sem[\confD]{\neg \predP} 
  = \text{not } \sem[\confD]{\predP}  
  \\[8pt]
  \sem[\confD]{\secrA} 
  = N 
  \;\;\text{if $\confD$ contains $\confRev{\pmvA}{\secrA}{N}$}
  \qquad
  \sem[\confD]{\expE[1] \bullet \expE[2]} 
  = \sem[\confD]{\expE[1]} \bullet \sem[\confD]{\expE[2]}
  \quad (\bullet \in \setenum{+,-})
  \\[8pt]
  \sem[\confD]{N} = N
  \qquad
  \sem[\confD]{\expE[1] \circ \expE[2]} 
  = \sem[\confD]{\expE[1]} \circ \sem[\confD]{\expE[2]}
  \quad (\circ \in \setenum{=,<})
\end{array}
\]

\mypar{Authorizing branches}

A branch $\authC{\pmvA}{\contrD}$ 
can be taken only provided that $\pmvA$ has granted her authorization.
This can be done through the following rule:
\[
\irule
{
}
{\confContr[\varX]{\authC{\pmvA}{\contrD} + \contrC}{\valV} \mid \confG
  \concrightarrow{\authLab{\pmvA}{(\varX,\authC{\pmvA}{\contrD})}}
  \confContr[\varX]{\authC{\pmvA}{\contrD} + \contrC}{\valV} 
  \mid \confAuth{\pmvA}{\authBranch{\varX}{\authC{\pmvA}{\contrD}}}
  \mid \confG
}
\smallnrule{[C-AuthBranch]}
\]

\mypar{Reducing branches}

Once all the authorizations for a branch occur in the configuration,
anyone can trigger the transition, provided that the time constraints
(if any) are respected.
For instance, if $\contrD = \authC{\pmvA}{\afterC{1000}{\withdrawC{\pmvB}}}$,
we have the transition:
\[
  \confContr[x]{\contrD + \contrC}
    {\valV}
    \mid
    \confAuth{\pmvA}{\authBranch{x}{\contrD}}
    \mid
    1051
    \concrightarrow{}\;
    \confDep[y]{\pmvB}{\valV}
    \mid
    1051
\]

\noindent
The general rule is the following:
\[
\irule
{
  \begin{array}{lll}
    \confContr[\varX]{\contrDi}{\valV} 
    \mid 
    \confG
    \concrightarrow{\labS} \confGi
    & \contrD = \authC{\pmvA[1]}{\authC{\cdots}{\authC{\pmvA[k]}{\afterC{\constT[1]}{\authC{\cdots}{\afterC{\constT[m]}{\contrDi}}}}}}
    & \contrDi \neq \authC{\pmvA}{\cdots}
    \\[3pt]
    \varX \not\in \cn{\confGi}
    & \constT \geq \constT[1],\ldots,\constT[m]
    & \contrDi \neq \afterC{\constTi}{\cdots}
  \end{array}
}
{
  \confContr[\varX]{\contrD + \contrC}{\valV} 
  \mid
  \mmid_{i=1}^k \confAuth{\pmvA[i]}{\authBranch{\varX}{\contrD}}
  \mid 
  \confG
  \mid 
  \constT
  \;\; \concrightarrow{\; \labS \;} \;\;
  \confGi
  \mid 
  \constT
}
\smallnrule{[C-Branch]}
\]

\mypar{Delaying}

In any configuration, we always allow time to advance:
\[ 
\irule
{\delta > 0}
{\confT{\confG}{\constT} \concrightarrow{\delta} \confT{\confG}{\constT+\delta}}
\smallnrule{[C-Delay]}
\]

\mypar{Renegotiation: advertisement}

Contract renegotiation is similar to stipulation, 
including advertisement, commitment, authorization, and contract initialization.
We illustrate these phases through a running example.
Consider a configuration:
\[
\confG[0] \; = \;
\confContr[\varX]{\adv{\callX{\constK}} + \contrC[{\textit{alt}}]}{\valV} \mid
\confG
\tag*{where $\callX{\procParamA} = \contrAdv{\contrG}{\contrC}$}
\] 
where $\contrC[{\textit{alt}}]$ contains the branches alternative to \rngtname,
\(
\contrG = 
\persdep{\pmvA}{1}{\varphX} \mid
\persdep{\pmvB}{1}{\varphY} \mid
\presecret{\pmvA}{\secrA}
\),
and $\contrC$ is an arbitrary contract involving only $\pmvA$ and $\pmvB$,
possibly containing the integer variable $\procParamA$ in static expressions.
The renegotiation is advertised as follows:
\begin{equation}
  \label{eq:bitml:renegotiation-advertisement}
  \confG[0]
  \; \concrightarrow{} \;
  \confG[0]
  \mid \contrAdv[\varX]{\contrGi}{\contrCi}
  \; = \; \confG[1]
\end{equation}
where $\contrAdv[\varX]{\contrGi}{\contrCi}$ is
obtained by transforming $\contrAdv{\contrG}{\contrC}$ as follows:
\begin{inlinelist}
\item variables $\varphX,\varphY$ are
  renamed into fresh ones $\varphXi,\varphYi$,
  and similarly the secret name $\secrA$ into $\secrAi$,
\item the static expressions in $\contrC$ are evaluated, assuming
  $\procParamA = \constK$, and replaced with their results.
\end{inlinelist}
The superscript $\varX$ in $\contrAdv[\varX]{\contrGi}{\contrCi}$ 
is used to record that,
when the renegotiation is concluded, the contract $\varX$ must be closed.
In general, this step is defined through the following rule:
\[
\irule
{
  \begin{array}{l}
    \contrAdv{\contrG}{\contrC} \equiv \callX{\vec{\sexp}}
    \\
    \text{all secrets in } \contrG \text{ fresh}
    \\
    \forall \persdep{\pmvA}{\valV}{\varphX} \text { in } \contrG :
    \varphX \text{ fresh }
  \end{array}
  \quad
  \begin{array}{l}
    \forall \contrAdv[\varX]{\contrGi}{\contrCi} \text{ in } \confG :
    \contrAdv{\contrGi}{\contrCi} \not\equiv \callX{\vec{\sexp}}
    \\[5pt]
    \forall \persdep{\pmvA[i]}{\valV[i]}{\varX[i]}
    \text{ in } \contrG :
    \confDep[{\varX[i]}]{\pmvA[i]}{\valV[i]} \text{ in } \confG
  \end{array}
}
{\begin{array}{l}
   \confContr[\varX]{\adv{\callX{\vec{\sexp}}} + 
   \contrC[{\textit{alt}}]}{\valV}
   \mid \confG
   \concrightarrow{{\it advRngt}(\contrAdv[\varX]{\contrG}{\contrC})}
   \\[3pt]
   \confContr[\varX]{\adv{\callX{\vec{\sexp}}} + \contrC[{\textit{alt}}]}{\valV} \mid
   \confG \mid
   \contrAdv[\varX]{\contrG}{\contrC} 
 \end{array}
}
\smallnrule{[C-AdvRngt]}
\]

\medskip
The relation $\contrAdv{\contrG}{\contrC} \equiv \callX{\vec{\sexp}}$ 
used in the rule premise holds when, for some $\contrGi$, $\contrCi$,
there exists a defining equation
$\decl{\cVarX}{\vec{\procParamA}} = \contrAdv{\contrGi}{\contrCi}$
such that $\contrAdv{\contrG}{\contrC}$ is the transformation of
$\contrAdv{\contrGi}{\contrCi}$ obtained by instantiating the formal parameters
with the actual ones, and by \mbox{$\alpha$-converting} the secret names 
and deposit variables, 
as done \eg in Equation~\eqref{eq:bitml:renegotiation-advertisement}.
Before formalising this transformation below
in Equation~\eqref{eq:bitml:renegotiation-advertisement:alpha},
we need a few auxiliary notions.
We denote with $\sem{\sexp}$ the evaluation of a 
closed static expression $\sexp$ 
(the actual definition is standard, so we omit it).
We overload $\sem{-}$ to contracts and contract advertisements: namely,
$\sem{\contrC}$ is the contract obtained by substituting all the occurring 
static expressions with their valuation, and 
$\sem{\contrAdv{\contrG}{\contrC}} = \contrAdv{\contrG}{\sem{\contrC}}$. 
Then, we denote with 
$\contrAdv{\contrG}{\contrC} \equiv_{\alpha} \contrAdv{\contrGi}{\contrCi}$
the $\alpha$-equivalence between two contract advertisements
\wrt secret names and deposit variables.
Finally, we define: 
\begin{equation}
  \label{eq:bitml:renegotiation-advertisement:alpha}
  \contrAdv{\contrG}{\contrC} \equiv \callX{\vec{\sexp}}
  \iff
  \exists \contrGi,\contrCi :
  \decl{\cVarX}{\vec{\procParamA}} = \contrAdv{\contrGi}{\contrCi}
  \text{ and }
  \sem{\contrAdv{\contrGi}{\contrCi}
    \subst{\vec{\sem{\sexp}}}{\vec{\procParamA}}} 
  \equiv_{\alpha} \contrAdv{\contrG}{\contrC}
\end{equation}

\mypar{Renegotiation: commitment}

In the subsequent steps participants choose the actual deposit names, 
and $\pmvA$ commits to her secret.
If $\pmvA$ owns in $\confG$ a deposit $\confDep[\varY]{\pmvA}{1}$,
she can choose $\varphXi = \varY$ to satisfy the precondition $\contrG$.
Similarly, $\pmvB$ can choose $\varphYi = \varZ$ if he owns such a deposit
in $\confG$.
These choices are performed as follows:
\begin{align*}
  \confG[1] 
  \; \concrightarrow{} \;
  \confG[1] 
  & \mid 
    \assign{\pmvA}{\varphXi}{\varY} \mid 
    \confSec{\pmvA}{\secrAi}{N} \mid \confAuth{\pmvA}{\authCommit[\varX]{\contrGi}
    {\contrCi}}
  && = \confG[2]
  \\ 
  \confG[2] 
  \; \concrightarrow{} \;
  \confG[2] 
  & \mid \assign{\pmvB}{\varphYi}{\varZ} \mid \confAuth{\pmvB}
    {\authCommit[\varX]{\contrGi}{\contrCi}}
  && = \confG[3]
\end{align*}

The general case is defined by the rule \nrule{[C-AuthCommit]}
described before.

\mypar{Renegotiation: authorization}

At this point, participants must authorise to spend their deposits
and the balance of the contract at $\varX$.
This is done through a series of steps:
\begin{align*}
  \confG[3] 
  & \concrightarrow{}
  \confG[3] 
  \mid \confAuth{\pmvA}{\authAdv[\varX]{\varY}{\contrGi}{\contrCi}}
  \\
  & \concrightarrow{}
  \confG[3] 
  \mid \confAuth{\pmvA}{\authAdv[\varX]{\varY}{\contrGi}{\contrCi}}
  \mid \confAuth{\pmvA}{\authAdv[\varX]{\varX}{\contrGi}{\contrCi}}
  && = \confG[4]
  \\ 
  & \concrightarrow{}
    \confG[4]
    \mid \confAuth{\pmvB}{\authAdv[\varX]{\varZ}{\contrGi}{\contrCi}} 
  \\ 
  & \concrightarrow{}
    \confG[4]
    \mid \confAuth{\pmvB}{\authAdv[\varX]{\varZ}{\contrGi}{\contrCi}} 
    \mid \confAuth{\pmvB}{\authAdv[\varX]{\varX}{\contrGi}{\contrCi}}
  && = \confG[5]
\end{align*}

\noindent
The general case is defined by the rule \nrule{[C-AuthInitDep]}
described above.

\mypar{Renegotiation: initialization}

Finally, the renegotiated contract is stipulated. 
This amounts to closing the old contract $\varX$, 
consuming the deposits $\varY$ and $\varZ$,
and transferring the balance $\valV$ of the old contract to the new one,
which is given a fresh name $\varXi$:
\[
\confG[5] 
\; \concrightarrow{} \;
\confContr[\varXi]{\contrCi}{\valV + 2} 
\mid \confGi
\]
where $\confGi$ is $\confG$ without the deposits $\varY$ and $\varZ$.
Note that the branches in $\contrC[{\textit{alt}}]$
are discarded only in the last step above, where we complete the renegotiation.
Before this step, it would have been possible to take one of the
branches in $\contrC[{\textit{alt}}]$, aborting the renegotiation.

The initialization step is defined by the following rule,
which is analogous to \nrule{[C-Init]}:
\[
\small
\irule{
  \begin{array}{l}
    \contrAdv{\contrG}{\contrC} \equiv \callX{\vec{\sexp}}
    \qquad
    \varY \text{ fresh}
    \\[3pt]
    \contrG = 
    \big( \mmid_{i \in I} \persdep{\pmvA[i]}{\valV[i]}{\varX[i]} \big) \mid
    \big( \mmid_{i \in J} \persdep{\pmvB[i]}{\valVi[i]}{\varphX[i]} \big) \mid
    \big( \mmid_{i \in K} \presecret{\pmvC[i]}{\secrA[i]} \big)
    \\[3pt]
    \confD = \big( \mmid_{i \in I} \confDep[{\varX[i]}]{\pmvA[i]}{\valV[i]} \big)
    \mid
    \big( \mmid_{i \in J} \confDep[{\varXi[i]}]{\pmvB[i]}{\valVi[i]} \big)
    \mid
    \big( \mmid_{i \in J} \assign{\pmvB[i]}{\varphX[i]}{\varXi[i]} \big)
    \mid
    \\[3pt]
    \hspace{22pt}
    \big( \mmid_{i \in I} \confAuth{\pmvA[i]}{\authAdv[\varX]{\varX[i]}{\contrG}{\contrC}} \big)
    \mid
    \big( \mmid_{i \in J} \confAuth{\pmvB[i]}{\authAdv[\varX]{\varXi[i]}{\contrG}{\contrC}} \big)
    \mid
    \\[3pt]
    \hspace{22pt}
    \big( \mmid_{\pmvA \in \contrG} \confAuth{\pmvA}{\authCommit[\varX]{\contrG}{\contrC}} \mid \confAuth{\pmvA}{\authAdv[\varX]{x}{\contrG}{\contrC}} \big)
  \end{array}
}
{\begin{array}{l}
   \confContr[\varX]{\adv{\callX{\vec{\sexp}}} + \contrC[{\textit{alt}}]}{\valV} \mid \contrAdv[\varX]{\contrG}{\contrC} \mid \confG \mid \confD
   \concrightarrow{{\it rngt}(\varX,\contrG,\contrC)}
   \\[3pt]
   \confContr[\varY]{\contrC}{\valV + \sum_{i \in I} \valV[i] + \sum_{i \in J} \valVi[i]}
   \mid \confG
 \end{array}
}
\smallnrule{[C-Rngt]}
\]

\medskip
The main difference between \nrule{[C-Init]} and \nrule{[C-Rngt]} 
is that the latter transfers the balance $\valV$ of the old contract 
to the new one.

\paragraph{Executing BitML on Bitcoin}

Stipulating or renegotiating a BitML contract $\contrC$ in Bitcoin requires each participant to invoke the BitML compiler, which was first introduced in~\cite{BZ18bitml}
and then extended in~\cite{BMZ20coordination} with renegotiation and recursion.
In particular, each participant has to 
\begin{inlinelist}
\item generate a key pair for each subcontract of $\contrC$,
\item exchange the generated public keys with the other participants,
\item sign each subcontract of $\contrC$, and
\item exchange all the signatures.
\end{inlinelist}
The whole protocol is detailed in~\cite{BZ18arxiv} (Definition 21).
When dealing with a contract $\contrC$ among $N$ participants and having $M$ subcontracts,
such protocol requires $O(N M)$ broadcasts.
After stipulation/renegotiation, executing a step of the BitML semantics corresponds to 
appending a transaction to the Bitcoin blockchain.

\section{Liquidity}
\label{sec:liquidity}

In this~\namecref{sec:liquidity} we formalise 
a notion of contract liquidity.
Aiming at generality, we parameterise this notion over: 
\begin{itemize}
\item an LTS $\rightarrow$, which models the contract behaviour;
\item a subset $\mathcal{L}_{\pmvA}$ of the labels of the LTS $\rightarrow$, 
  which represents the moves that can be performed by an \emph{honest} 
  participant $\pmvA$, without requiring the cooperation of the other participants.
\end{itemize}

Once these parameters are fixed, we define when,
in a configuration $\confG$ of the LTS,
a set $\VarX$ of contracts in $\confG$ is liquid.
Roughly, this happens when the honest participant $\pmvA$
can always make the funds stored by the contracts $\VarX$ 
be transferred to some participant. 
In the meanwhile, the other participants may play against her, 
\eg by not revealing some secrets, 
or by not granting their authorizations for some branch.
Note that by suitably instantiating the parameters $\rightarrow$
and $\mathcal{L}_{\pmvA}$, we will be able to use 
the same notion of liquidity
both with the concrete and with the abstract BitML semantics.
For instance, for the concrete semantics we choose
$\rightarrow \, = \, \concrightarrow{}$,
and as $\mathcal{L}_{\pmvA}$ all the labels except those of the form
$\pmvB:\ell$ with $\pmvB \neq \pmvA$.

We start by introducing an auxiliary partial function 
$\origin[\confG]{\runS}{\varX}$ that,
given a contract name $\varX$
and a run $\runS$ starting from $\confG$,
determines the (unique) ancestor $\varY$ of $\varX$ in $\confG$, if any.
Intuitively, $\origin[\confG]{\runS}{\varX} = \varY$ means that 
$\varY$ has evolved along the run $\runS$, 
eventually leading to $\varX$, and possibly to other contracts.
In BitML, this happens in one of the following cases.
First, a $\splitC{}{}$ can spawn new contracts, \eg:
\[
\confContr[\varX]
{\splitC{(\splitB{\valV[1]}{\contrC[1]} \mid \splitB{\valV[2]}{\contrC[2]})}}
{\valV[1]+\valV[2]} 
\concrightarrow{{\it split}(\varX)}
\confContr[{\varY[1]}]{\contrC[1]}{\valV[1]} \mid 
\confContr[{\varY[2]}]{\contrC[2]}{\valV[2]}
\]
Here, both $\varY[1]$ and $\varY[2]$ have $\varX$ as ancestor.
Second, $\revealname$ reduces as follows:
\[
\confContr[\varX]{\putC{}{\secrA} . \, \contrC}{\valV} \mid \cdots
\;\concrightarrow{{\it reveal}(\secrA,\varX)}\;
\confContr[\varY]{\contrC}{\valV} \mid \cdots
\]
In this case, the ancestor of $\varY$ is $\varX$.
Third, a $\rngtname$ evolves as follows:
\[
\confContr[\varX]{\adv{\callX{\vec{\sexp}}}}{\valV} \mid \cdots
\concrightarrow{{\it init}(\varX,\contrG,\contrC)}
\confContr[\varY]{\contrC}{\valV + \cdots}
\mid \cdots
\]
Also in this case, the ancestor of $\varY$ is $\varX$.

\begin{defi}[Origin]
  \label{def:origin}
  Let $\runS$ be a run starting from $\confG$,
  and let $\varX$ be a contract name.
  We define the partial function $\origin[\confG]{\runS}{\varX}$ 
  by induction on the length of $\runS$ in~\Cref{fig:origin}.
\end{defi}

Note that, in~\Cref{fig:origin},
the condition $\varX \in \cn{\confGi} \setminus \cn{\confS{\runS}}$
checks that the name $\varX$ has been introduced in the last transition
of the run.
Dually, $\setenum{\varY} = \cn{\confS{\runS}} \setminus \cn{\confGi}$
checks that $\varY$ has been consumed.
The use of a singleton $\setenum{\varY}$ is justified by the fact that in BitML
each transition can consume at most one contract name.

\begin{figure}[t]
  \small
  \[
  \begin{array}{rcl}
    \origin[\confG]{\confG}{\varX} 
    & = & x \quad \text{if $\varX \in \cn{\confG}$}
    \\[4pt]
    \origin[\confG]{\runS \xrightarrow{\labS} \confGi}{\varX} 
    & = & \begin{cases}
      \origin[\confG]{\runS}{\varX} 
      & \text{if $\,\varX \in \cn{\confS{\runS}} \cap \cn{\confGi}$}
      \\
      \origin[\confG]{\runS}{\varY} 
      & \text{if $\varX \in \cn{\confGi} \setminus \cn{\confS{\runS}}$ and
        $\setenum{\varY} = \cn{\confS{\runS}} \setminus \cn{\confGi}$}
    \end{cases}
  \end{array}
  \]
  \caption{Origin of a contract name within a run.}
  \label{fig:origin}
\end{figure}

\begin{exa}
  \label{ex:origin}
  Let $\confG = \confContr[y]{\contrC[1]}{v} \mid \confDep[{\varZ}]{\pmvA}{\valV}$,
  and let $\runS$ be the following run starting from $\confG$,
  where the contracts $\contrC[1]$ and $\contrC[2]$ are immaterial,
  but for the fact that they enable the displayed moves:
  \begin{align*}
    \confContr[y]{\contrC[1]}{\valV} \mid \confDep[{\varZ}]{\pmvA}{\valV}
    & \concrightarrow{} 
      \confContr[y]{\contrC[1]}{\valV} \mid \confDep[{\varZ}]{\pmvA}{\valV} \mid \contrAdv{\contrG}{\contrC[2]} 
      \concrightarrow{}^*
      \confContr[y]{\contrC[1]}{\valV} \mid \confContr[x]{\contrC[2]}{\valV} 
    \\
    & \concrightarrow{{\it split}(x)} 
      \confContr[y]{\contrC[1]}{\valV} \mid \confContr[x']{\contrCi[2]}{\valV}
    \\
    & \concrightarrow{{\it split}(y)} 
      \confContr[y']{\contrCi[1]}{\valVi} \mid \confContr[y'']{\contrCii[1]}{\valV-\valVi} \mid
      \confContr[x']{\contrCi[2]}{\valV}
  \end{align*}
  We have that
  $\origin[\confG]{\runS}{y'} = \origin[\confG]{\runS}{y''} = y$,
  since the corresponding contracts have been obtained through a split of
  the ancestor $y$, which was in $\confG$.
  Instead, $\origin[\confG]{\runS}{x'}$ is undefined,
  because its ancestor $x$ is not in $\confG$.
  Further, $\origin[\confG]{\runS}{y} = y$, 
  while $\origin[\confG]{\runS}{x}$ is undefined.
\end{exa}

Before formalising liquidity, we give some further intuition.
Assume that $\pmvA$ is an honest participant, 
who cares about the liquidity of a set of contracts $\VarX$ in $\confG$.
After an arbitrary sequence of transitions 
$\confG \rightarrow \cdots \rightarrow \confGi$,
where any participant may perform actions,
$\pmvA$ wants to liquidate all the contracts in $\confGi$
originating from $\VarX$, transferring their funds to 
participants' deposits.
We want $\pmvA$ to be able to liquidate contracts \emph{without} 
the help of the other participants.

\begin{defi}[Liquidity]
  \label{def:liquid}
  Let $\rightarrow$ be an LTS, let $\mathcal{L}_{\pmvA}$ be a subset of 
  its labels, let $\confG$ be a configuration of the LTS, 
  and let $\varX$ be a contract name.
  We say that:
  \begin{itemize}

  \item $\varX$ is \emph{liquidable} in $\confG$ if there exists a run
    \(
    \runS \; = \; \confG \xrightarrow{\ell_1} \cdots \xrightarrow{\ell_n} \confGi
    \)
    such that:
    \begin{enumerate}
    \item for all $i \in 1..n$, $\ell_i \in \mathcal{L}_{\pmvA}$;
    \item \label{item:liquidability:cn}
      there exists no $\varZ \in \cn{\confGi}$ such that
      $\origin[\confG]{\runS}{\varZ} = \varX$.
    \end{enumerate}

  \item $\varX$ is \emph{liquid} in $\confG$
    if, for all runs $\runS = \confG \xrightarrow{} \cdots \xrightarrow{} \confGi$,
    all the contract names $\varY$
    with $\origin[\confG]{\runS}{\varY} = \varX$
    are liquidable in $\confGi$.

  \end{itemize}
  We extend this to sets $\VarX$ of contracts names:
  $\VarX$ is liquid in $\confG$ iff all $\varX \in \VarX$ are liquid in $\confG$.
\end{defi}

Intuitively, a contract $\varX$ is liquidable by $\pmvA$
when $\pmvA$ can perform a sequence of transitions which
eventually lead to a configuration containing no contract names 
originated from $\varX$.
Consequently, all the funds in the contract $\varX$ have been 
transferred to participants' deposits.
When liquidating $\varX$, 
$\pmvA$'s moves can not reveal secrets of other participants, 
or generate authorizations for them:
$\pmvA$ must be able to unfreeze the funds on her own,
performing actions $\ell_i \in \mathcal{L}_{\pmvA}$.
Note that if $\varX \not\in \cn{\confG}$, then 
$\varX$ is trivially liquidable in $\confG$.

The notion of liquidity is based upon that of liquidability.
A contract $\varX$ is liquid in a configuration $\confG$
when, after an arbitrary sequence of moves
performed by \emph{any} participant,
the contract names originated by $\varX$ are liquidable by $\pmvA$.

We remark that, although~Definition~\ref{def:liquid} 
will be instantiated with the semantics of BitML,
the basic concepts it relies upon 
(runs, origin of contracts, moves of a participant)
are quite general. 
Hence, our notion of liquidity, 
as well as the variants proposed in Section~\ref{sec:variants-liquidity}, 
can be applied to other languages for smart contracts,
using their transition semantics.

\begin{exa}
  Recall the timed commitment contract $\TCommitment$ from Section~\ref{sec:intro}:
  \[
  \contrAdv
  {\persdep{\pmvA}{1}{\varZ} \mid \presecret{\pmvA}{\secrA}}
  {\;
    ( 
    \putC{}{\secrA} . \, \withdrawC{\pmvA}
    \;\; + \;\;
    \afterC{\constT}{\withdrawC{\pmvB}}
    )
  }
  \]
  Let
  \(
  \confG = 
  \confContr[x]{\TCommitment}{1} \mid \confSec{\pmvA}{\secrA}{N} 
  \)
  be a configuration where the contract has been stipulated.
  We show that $\varX$ is liquid in $\confG$ with respect to any participant.
  In the configurations reachable from $\confG$,
  the contract $\TCommitment$ has not progressed,
  or it has reduced to $\withdrawC{\pmvA}$.
  In the first case, $\varX$ is liquidable by anyone, 
  by firing $\withdrawC{\pmvB}$ after time $\constT$ 
  (alternatively, $\pmvA$ may reveal the secret and 
  then fire $\withdrawC{\pmvA}$).
  Instead, if the contract has reduced to $\withdrawC{\pmvA}$,
  anyone can liquidate it.
  Since all the descendants of $\varX$ are liquidable, $\varX$ is liquid.
\end{exa}

\begin{exa}
  Let $\contrG = \persdep{\pmvA}{1}{\varY} \mid \persdep{\pmvB}{1}{\varZ} \mid
  \presecret{\pmvA}{\secrA}$.
  Consider the following contracts, 
  where $\predP$ is an arbitrary predicate on $\secrA$:
  \begin{align*}
    \contrC[1] 
    & = \putC[\predP]{}{\secrA} . \, \withdrawC{\pmvA}
      + \putC[\neg \predP]{}{\secrA} . \, \withdrawC{\pmvB}
    \\
    \contrC[2]
    & = \splitC{\;
      \big(
      \splitB{1}{\putC[\predP]{}{\secrA} . \, \withdrawC{\pmvA}}}
      \;\mid\;
      \splitB{1}{\putC[\neg \predP]{}{\secrA} . \, \withdrawC{\pmvB}} 
      \big)
    \\
    \contrC[3]
    & = \withdrawC{\pmvA} + \authC{\pmvB}{\splitC{\big( \splitB{1}{\authC{\pmvB}{\withdrawC{\pmvA}}} \mid \splitB{1}{\withdrawC{\pmvB}}} \big)}
  \end{align*}
  For $i \in \setenum{1,2,3}$, let 
  $\confG[i] = \confContr[x]{\contrC[i]}{2} \mid \confSec{\pmvA}{\secrA}{N}$, 
  with $N \neq \bot$.  
  We have that:
  \begin{itemize}
  \item $\varX$ is liquid in $\confG[1]$ for $\pmvA$,
    but not for any other participant.
    $\contrC[1]$ has three reducts: $\contrC[1]$ itself,
    $\withdrawC{\pmvA}$, and $\withdrawC{\pmvB}$.
    The last two contracts are trivally liquidable by anyone.
    Instead $\contrC[1]$ is liquidable only by $\pmvA$,
    by revealing $\secrA$ and firing a \revealname\, branch
    (since their guards are $\predP$ and $\neg \predP$, one of them 
    will be enabled),
    and finally firing the corresponding $\withdrawC{}$.
    Instead, from $\pmvB$'s viewpoint $\varX$ is not liquid,
    because $\pmvA$ could refuse to reveal.

  \item $\varX$ is not liquid in $\confG[2]$ for anyone.
    Indeed, none of the reducts of $\contrC[2]$ is liquidable,
    because one of the two $\revealname$ branches is stuck.

  \item $\varX$ is liquid in $\confG[3]$ for $\pmvB$, but not for $\pmvA$.
    Indeed, if $\pmvB$ authorizes and performs the $\splitname$,
    then the reduct is not liquidable by $\pmvA$, since $\pmvB$
    could deny his second authorization.
  \end{itemize}
\end{exa}

Note that liquidability in Definition~\ref{def:liquid} 
requires that the moves $\labS[1] \ldots \labS[n]$ are performed atomically,
effectively forbidding the adversary to nterfere.
Atomicity might be realistic in some blockchains, but not in others. 
For instance, 
Algorand features atomic sequences of transactions natively~\cite{BB+21fc},
while in Ethereum it is possible to perform atomically
a sequence $\labS[1] \ldots \labS[n]$ of calls  
by deploying a new contract with a suitable function which 
calls $\labS[1] \ldots \labS[n]$ in sequence.
Bitcoin, instead, does not support atomic sequences of transactions: 
an honest participant could start to perform the sequence, 
but at some point in the middle the adversary can interfere. 
Repeated interference could lead to an infinite run, 
where each attempt by the honest participant is hindered by the adversary.
To illustrate the issue, consider a (not-BitML) LTS 
with states 0, 1, 2 and transitions
$0 \xrightarrow{s} 1$,
$1 \xrightarrow{s} 2$,
and $1 \xrightarrow{p} 0$,
where 0 is the initial state and 2 is the final state, where the funds have been liquidated.
If atomic sequences of moves were allowed, 
an honest participant could always reach the final state by firing the atomic sequence $s \ s$.
Otherwise, an adversary could always prevent the participant from reaching the final state,
by firing $p$ after each $s$.
Hence, this LTS would be considered liquid only by assuming atomic sequences of moves.

However, infinite adversary interference as those shown above
are not reproducible in BitML, for the following reason.
Our notion of liquidity requires that any descendent 
of the contract must be liquidable by $\pmvA$ \emph{alone}. 
This means that $\pmvA$ can do that by performing a sequence of moves
which do not include any renegotiation, since renegotiations 
can be finalized only with the cooperation of all the other participants.
Without renegotiation, the contract eventually terminates,
so an adversary can interfere at most a finite number of times.
After the last interference, $\pmvA$ has still a way of 
terminating the contract, and she can do this alone,
satisfying item~\eqref{item:liquidability:cn} of Definition~\ref{def:liquid}. 

\section{Case studies}
\label{sec:examples}
\label{sec:cfg}

In this~\namecref{sec:examples} we illustrate BitML and liquidity 
through a few example contracts. 
The automatic verification of these contracts will be discussed later on
in Section~\ref{sec:toolchain}.

\paragraph{A fair lottery}

\newcommand{\Lottery}[1][]{\ensuremath{\contrFmt{\it Lottery}\ifempty{#1}{}{(#1)}}\xspace}
\newcommand{\Win}[1][]{\ensuremath{\contrFmt{\it Win}_{\ifempty{#1}{}{\contrFmt{#1}}}}\xspace}

Consider a lottery between two players.
The preconditions require $\pmvA$ and $\pmvB$ to 
commit to one secret each ($\secrA$ and $\secrB$, respectively),
and to put a deposit of $3 \BTC$ each
($1\BTC$ as a bet, and $2\BTC$ as a penalty for dishonest behaviour):
\begin{align*}
  & \Lottery[\Win]
    = 
    \splitC{}  \big(
  \\[-2pt]
  & \hspace{25pt}
    \splitB{2}{}
    ( \putC[0 \leq \, \secrB \leq 1]{}{\secrB} . \; \withdrawC{\pmvB} )
    \; + \; 
    ( \afterC{\constT}{\withdrawC{\pmvA}} )
  \\[-2pt] 
  & 
    \hspace{16pt}
    \mid 
    \splitB{2}{}
    ( \putC{}{\secrA} . \; \withdrawC{\pmvA} )
    \; + \; 
    ( \afterC{\constT}{\withdrawC{\pmvB}} )
  \\[-2pt]
  & \hspace{16pt}
    \mid  
    \splitB{2}{\Win}
    \big)
  \\
  & \Win = \putC[\secrA=\secrB]{}{\secrA \, \secrB} . \; \withdrawC{\pmvA}
    \; + \;
    \putC[\secrA \neq \secrB]{}{\secrA \, \secrB} . \; \withdrawC{\pmvB}
\end{align*}
The contract splits the balance in three parts, of $2\BTC$ each.
The first part allows $\pmvB$ to reveal $\secrB$ and then redeem $2\BTC$;
otherwise, after the deadline $\pmvA$ can redeem $\pmvB$'s penalty
(as in the timed commitment).
Similarly, the second part allows $\pmvA$ to redeem $2\BTC$
by revealing $\secrA$.
To determine the winner we compare the secrets, in the subcontract $\Win$:
$\pmvA$ wins if the secrets are the same, 
otherwise $\pmvB$ wins.
This lottery is \emph{fair}, since:
\begin{inlinelist}
\item if both players are honest,
  then they will reveal their secrets within the deadline 
  (redeeming $2\BTC$ each),
  and then they will have a $1/2$ probability of winning%
  \footnote{Note that $\pmvB$ could increase his probability 
    to win the lottery by choosing a secret $\secrB>1$,
    since doing so would increase the chances that $\secrA \neq \secrB$. 
    We require that $0 \leq \secrB \leq 1$, so that
    if $\pmvB$ chooses a secret outside that range
    he will lose his $2\BTC$ deposit
    in the first part of $\splitname$, and so 
    $\pmvB$'s \emph{average} payoff would be negative.
    Instead, $\pmvA$ can only decrease decrease her probability to win
    by choosing a secret $\secrA \not\in \setenum{0,1}$.
    For this reason, it is not necessary to require that $0 \leq \secrA \leq 1$.};
\item if a player is dishonest, not revealing the secret,
  then the other player has a positive payoff,
  since she can redeem $4\BTC$.
\end{inlinelist}

Although fair, the lottery is \emph{not} liquid,
neither from $\pmvA$'s nor from $\pmvB$'s point of view,
because if one of the two players does not reveal her secret, 
then the $\Win$ subcontract is not liquidable,
and so the $2\BTC$ stored therein are frozen.
We can recover liquidity by replacing $\Win$ with the following contract,
where $\constTi > \constT$:
\begin{align*}
  & \Win[2] 
    \; = \; \Win 
    + (\afterC{\constTi}{\putC{}{\secrA} . \; \withdrawC{\pmvA}})
    + (\afterC{\constTi}{\putC{}{\secrB} . \; \withdrawC{\pmvB}})
\end{align*}

\noindent
Now, even if one of the two players does not reveal, 
the honest player fire her \revealname\; at time $\constTi$,
liquidating the $2\BTC$ stored in $\Win[2]$.

\paragraph{Zero-coupon bonds}

A zero-coupon bond~\cite{PeytonJones00icfp} is a financial contract
where an investor $\pmvA$ pays $1 \BTC$ upfront to a bank $\pmvB$,
and receives back $2 \BTC$ after a maturity date (say, year 2030).
We can express this contract in BitML as follows.
The contract precondition requires $\pmvA$ to provide a deposit $\varX$
of $1 \BTC$, and $\pmvB$ to provide a deposit of $2 \BTC$.
The contract is as follows:
\begin{align*}
  \ZCB 
  & \; = \;
    \splitname \; \big( 
    \splitB{1}{\withdrawC{\pmvB}}
    \; \mid \;
    \splitB{2}{\afterC{\text{2030}}{\withdrawC{\pmvA}}}
    \big)
\end{align*}

Upon stipulation, all the deposits required in the preconditions
pass under the control of $\ZCB$, 
and can no longer be spent by $\pmvA$ and $\pmvB$.
The contract splits these funds in two parts:
$1\BTC$, that can be withdrawn by $\pmvB$ at any moment,
and $2\BTC$, that can be withdrawn by $\pmvA$ after the maturity date.

Although $\ZCB$ correctly implements the functionality of zero-coupon bounds, 
it is quite impractical:
for the whole period from the stipulation to the maturity date,
$2\BTC$ are frozen within the contract, and cannot be used by the bank
in any way.
Although this is a desirable feature for the investor,
since it guarantees that he will receive $2\BTC$ even if the bank bankrupts,
it is quite undesirable for the bank.
In the real world, the bank would be free to use its own funds,
together with those of investors, 
to make further financial transactions through which to repay the investments.
The risk that the bank bankrupts is mitigated
by external mechanisms, like insurances or government intervention.

To overcome this issue, we can exploit renegotiation.
We first revise the precondition, which now requires only $\pmvA$'s deposit. 
The revised contract is:
\begin{align*}
  \ZCBi
  & \; = \; \splitname \; \big( 
    \splitB{1}{\withdrawC{\pmvB}}
    \; \mid \;
    \splitB{0}{\adv{\callX{}}}
    \big)
  \\
  \callX{} 
  & \;= \; \contrAdv{\persdep{\pmvB}{2}{\varphX}}{\;\afterC{\text{2030}}{\withdrawC{\pmvA}}}
\end{align*}

As before, the bank can withdraw $1\BTC$ at any moment after stipulation.
In the second part of the $\splitname$, 
the participants renegotiate the contract:
if they both agree, $0\BTC$ pass under the control of the contract $\callX{}$.
The precondition of $\callX{}$ requires the bank 
to provide $2\BTC$ in a fresh deposit;
upon renegotiation, $\pmvA$ can withdraw $2\BTC$ after the maturity date.
The crucial difference with $\ZCB$ is that the deposit variable $\varphX$
is instantiated at \emph{renegotiation} time, unlike $\varX$,
which must be fixed at \emph{stipulation} time.

The revised contract $\ZCBi$ solves the problem of $\ZCB$, 
in that it no longer freezes $2\BTC$ for the whole duration of the bond:
the bank could choose to renegotiate the contract,
paying $2\BTC$, just before the maturity date.
This flexibility comes at a cost, since $\pmvA$ loses
the guarantee to eventually receive $2\BTC$.
To address this issue we need to add, as in the real world,
an external mechanism.
More specifically, we assume an insurance company $\pmv{I}$ that,
for an annual premium of $p \BTC$ paid by the bank, 
covers a face amount of $f \BTC$ (with $2 > f > 10 p$):
\[
\persdep{\pmvA}{1}{\varX[1]} 
\mid
\persdep{\pmvB}{p}{\varX[2]}
\mid
\persdep{\pmv{I}}{f}{\varX[3]}
\]
We revise the bond contract as follows:
\begin{align*}
  \ZCBii
  & \; = \; \splitname \; \big( 
    \splitB{1}{\withdrawC{\pmvB}}
  \\
  & \hspace{48pt} \mid\hspace{1pt}
    \splitB{p}{\withdrawC{\pmv{I}}}
  \\
  & \hspace{48pt} \mid
    \splitB{f}{\adv{\callX{1}} + \afterC{2021}{\withdrawC{\pmvA}}}
    \big)
  \\
  \callX{n \in 1..9} 
  & \;= \; \contrAdv{\persdep{\pmvB}{p}{\varphX}}{}
  \\
  & \hspace{22pt} \splitname \; \big( \splitB{p}{\withdrawC{\pmv{I}}}
  \\
  & \hspace{48pt} \mid
  \splitB{f}{\adv{\callX{n+1}} + \afterC{(2021+n)}{\withdrawC{\pmvA}}}
  \big)
  \\
  \callX{10} 
  & \;= \; \contrAdv{\persdep{\pmvB}{2}{\varphX}}{}
  \\
  & \hspace{22pt} \splitname \; \big( \splitB{f}{\withdrawC{\pmv{I}}}
  \\
  & \hspace{48pt} \mid
    \splitB{2}{\afterC{\text{2030}}{\withdrawC{\pmvA}}}
    \big)
\end{align*}

The contract starts by transferring $1\BTC$ to the bank,
and the first year of the premium to the insurer.
The remaining $f\BTC$ are transferred to the renegotiated contract
$\callX{1}$, or, if the renegotiation is not completed by 2021,
to the investor.

The contracts $\callX{n}$, for $n \in 1..9$, 
allow the insurer to receive the annual premium until 2030:
if the bank does not renegotiate the contract for the following year
(paying the corresponding premium), then
the investor can redeem the face amount of $f \BTC$.
Finally, the contract $\callX{10}$ can be triggered if the bank deposits
the $2\BTC$: when this happens, the face amount is given back to the insurer,
and the investor can redeem $2\BTC$ after the maturity date.
 
Compared to $\ZCBi$, the contract $\ZCBii$ offers more protection 
to the investor.
To see why, we must evaluate $\pmvA$'s payoff for all the possible 
behaviours of the other participants.
If $\pmvB$ and $\pmv{I}$ are both honest,
then $\pmvA$ will redeem $2 \BTC$, as in the ideal contract $\ZCB$.
Instead,
if either $\pmvB$ or $\pmv{I}$ do not accept to renegotiate some $\callX{n}$, 
then $\pmvA$ can redeem $f \BTC$ as a partial compensation
(unlike in $\ZCBi$, where $\pmvA$ just loses $1\BTC$).
In the real world, $\pmvA$ could use this compensation
to cover the legal fee to sue the bank in court;
also, $\pmv{I}$ could \eg increase the premium for future interactions 
with $\pmvB$.
By further refining the contract, 
we could model these real-world mechanisms as oracles, 
which sanction dishonest participants
according to the evidence collected in the blockchain
and in messages broadcast by participants.
For instance, if $\pmvB$ and $\pmv{I}$ accept the renegotiation $\callX{n}$
but $\pmvA$ does not, 
then the oracle would be able to detect $\pmvA$'s dishonesty
by inspecting the authorizations broadcast in year $2021+n$.
The sanction could consist \eg in preventing $\pmvA$ from buying other bonds from $\pmvB$.

Despite the apparent complexity, the contract is liquid for any participant:
indeed, any potentially blocking renegotiation has an alternative
timeout branch, which allows to liquidate the contract.

\paragraph{A fair recursive coin flipping game}

\newcommand{\CFG}{\ensuremath{\contrFmt{\it CFG}}\xspace}
\newcommand{\Split}[1][]{\ensuremath{\contrFmt{{\it Split}_{#1}}}\xspace}
\newcommand{\CFGindent}{\hspace{40pt}}
\newcommand{\CFGindentNP}{\hspace{30pt}}
\newcommand{\CFGtab}{\hspace{10pt}}

Consider a simple game where two players repeatedly flip coins,
and the one who wins two consecutive flips takes the pot.
The precondition requires each player to deposit $3\BTC$ and choose a secret:
\begin{align*}
  \persdep{\pmvA}{3}{x} \;\mid\; \presecret{\pmvA}{\secrA}
  \; \mid \;
  \persdep{\pmvB}{3}{y} \;\mid\; \presecret{\pmvB}{\secrB}
\end{align*}

\begin{figure}[t]
  \small
  \[
  \begin{array}{l}
    \contrFmt{\it CFG}
    =
    \hspace{6pt}
    \putC[0 \leq \secrB \leq 1]{}{\secrB} . \big( 
    \\
    \CFGindent\CFGtab
    \putC[\secrA = \secrB]{}{\secrA\secrB}. \;
    (
    \adv{\call{\cVarX[\pmvA]}{1}}
    \, + \,
    \afterC{3}{\Split[\pmvA]}
    )
    \\
    \CFGindent
    + \;\putC[\secrA \neq \secrB]{}{\secrA\secrB} . \; 
    (
    \adv{\call{\cVarX[\pmvB]}{1}}
    \, + \,
    \afterC{3}{\Split[\pmvB]}
    )
    \\ 
    \CFGindent
    + \;
    \afterC{2}{\withdrawC{\pmvB}}
    \big)
    \\
    \CFGindentNP
    + \; \afterC{1}{\withdrawC{\pmvA}}
    \\[4pt]
    \call{\cVarX[\pmvA]}{n} 
    = \contrAdv{\presecret{\pmvA}{\secrA} \mid 
    \presecret{\pmvB}{\secrB}}{} 
    \\
    \CFGindent
    \putC[0 \leq \secrB \leq 1]{}{\secrB} . \big( 
    \\
    \CFGindent\CFGtab
    \putC[\secrA = \secrB]{}{\secrA\secrB} . \; \withdrawC{\pmvA}
    \\
    \CFGindent
    + \;\putC[\secrA \neq \secrB]{}{\secrA\secrB} . \;
    (
    \adv {\call{\cVarX[\pmvB]}{n + 1}} 
    \, + \,
    \afterC{(3n+3)}{\Split[\pmvB]}
    )
    \\ 
    \CFGindent + \;
    \afterC{(3n+2)}{\withdrawC{\pmvB}} \big)
    \\
    \CFGindentNP
    + \; \afterC{(3n+1)}{\withdrawC{\pmvA}}
    \\[4pt]
    \call{\cVarX[\pmvB]}{n} 
    = \contrAdv{\presecret{\pmvA}{\secrA} \mid 
    \presecret{\pmvB}{\secrB}}{}
    \\
    \CFGindent
    \putC[0 \leq \secrB \leq 1]{}{\secrB} . \big( 
    \\
    \CFGindent\CFGtab
    \putC[\secrA = \secrB]{}{\secrA\secrB} . \;
    (
    \adv{\call{\cVarX[\pmvA]}{n + 1}}
    \; + \;
    \afterC{(3n+3)}{\Split[\pmvA]}
    )
    \\
    \CFGindent
    + \;\putC[\secrA \neq \secrB]{}{\secrA\secrB} . \; 
    \withdrawC{\pmvB}
    \\ 
    \CFGindent + \;
    \afterC{(3n+2)}{\withdrawC{\pmvB}} \big)
    \\
    \CFGindentNP
    + \; \afterC{(3n+1)}{\withdrawC{\pmvA}}
    \\[4pt]
    \Split[\pmvA] 
    =
    \splitC{(
    \splitB{4}{\withdrawC{\pmvA}} \mid 
    \splitB{2}{\withdrawC{\pmvB}})}
    \\[4pt]
    \Split[\pmvB] 
    =
    \splitC{(
    \splitB{4}{\withdrawC{\pmvB}} \mid 
    \splitB{2}{\withdrawC{\pmvA}})}
  \end{array}
  \]
  \caption{A recursive coin flipping game.}
  \label{fig:cfg}
\end{figure}

The contract \CFG (\Cref{fig:cfg})
asks $\pmvB$ to reveal his secret first:
if $\pmvB$ waits too much, $\pmvA$ can withdraw
the contract funds after time 1.
Then, it is $\pmvA$'s turn to reveal 
(before time 2, otherwise $\pmvB$ can withdraw the funds).
The current flip winner is $\pmvA$
if the secrets of $\pmvA$ and $\pmvB$ are equal, 
otherwise it is $\pmvB$.
At this point, the contract can be renegotiated as 
$\call{\cVarX[\pmvA]}{1}$ or $\call{\cVarX[\pmvB]}{1}$, 
depending on the flip winner
(the parameter $1$ represents the round).
If players do not agree on the renegotiation,
then the funds are split fairly, according to the current expected win.
The contract $\call{\cVarX[\pmvA]}{n}$ requires $\pmvA$ and $\pmvB$
to generate fresh secrets for the $n$-th round.
If $\pmvA$ wins again, she can withdraw the pot, 
otherwise the contract can be renegotiated as $\call{\cVarX[\pmvB]}{n+1}$.
If the players do not agree on the renegotiation,
the pot is split fairly between them. 
The contract $\cVarX[\pmvB]$ is similar.

This game is fair, \ie
the expected payoff of a rational player is always non-negative,
notwithstanding the behaviour of the other player.
Rational players must choose random secrets in $\setenum{0,1}$,
since non uniformly distributed secrets 
can make the adversary bias the coin flip in her favour.
Further, choosing a secret different from $0$ or $1$ would be irrational:
if done by $\pmvB$, this would prevent himself from revealing
(by the predicate in the $\putC{}{\secrB}$), 
and so $\pmvA$ could win after the timeout;
if done by $\pmvA$, this would make $\pmvB$ win the round
(since $\pmvB$ wins when the secrets are different).
Rationality also requires to reveal secrets in time 
(before the alternative $\aftername$ branch is enabled),
and to take the $\Split$ branch if restipulation does not occur in time.
This ensures that, when renegotiation happens, 
there is still time to reveal the round secrets.
Indeed, a late renegotiation could enable the other player 
to win by timeout.
To show fairness,
first consider the case where renegotiation always happens.  
A rational player wins each coin flip with probability $1/2$, at
least: so, the probability of winning the whole game is also $1/2$,
at least.  In the general case, the renegotiation at the end of each
round may fail. When this happens, the rational player takes the
$\Split$ branch, distributing the pot according to the expected
payoff in the \emph{current} game state, thus ensuring the fairness
of the game. 
The player who won the last coin flip is expected to win $p\BTC$, with
$p = \nicefrac{1}{2}\cdot 6 +
\nicefrac{1}{2}\cdot(\nicefrac{1}{2}\cdot p + \nicefrac{1}{2}\cdot
0)$, giving $p=4$.  Accordingly, the $\Split$ contracts transfer
$4\BTC$ to the winner of the last flip and $(6-4)\BTC = 2\BTC$ to
the other player.

The contract $\CFG$ is liquid, because the $\aftername$ branches 
always offer liquidable alternatives to potentially blocking branches.

\section{A safe abstraction of the BitML  semantics}
\label{sec:abs-bitml}

The concrete BitML semantics is infinite-state, 
because participants can always create new deposits,
stipulate new contracts and renegotiate them, 
and can advance the current time.
In this~\namecref{sec:abs-bitml} 
we introduce an abstract semantics of BitML,
which reduces the state space to a finite one,
safely approximating liquidity.
We construct our abstraction in three steps:
\begin{itemize}

\item First, we abstract concrete configurations $\confG$
  as \emph{abstract configurations} $\absConf[\pmvA,\VarX]{\confG}$,
  where $\pmvA \in \PartT$, and
  $\VarX$ is the (finite) set of contract names under observation.
  Roughly, $\absConf[\pmvA,\VarX]{\confG}$ discards all the terms in the configuration,
  except for the contracts $\VarX$, which are abstracted as follows.
  We remove $\revealname$ actions, authorizations, 
  and time constraints, only recording whether reducing a contract $\contrD$
  requires cooperation from some participant different from $\pmvA$.
  Further, the abstraction discards the actual parameters $\sexp$ of renegotiations.

\item We define a semantics $\absrightarrow{}$ of abstract configurations.
  This semantics partitions the moves in two sets:
  the moves $\varX$, which represent a reduction of the contract $\varX$ that can be performed by $\pmvA$ \emph{alone}, 
  and the moves $\advPref{\varX}$, which instead represent actions that may
  require cooperation from other participants.
  Intuitively, the moves of the first kind are those required by liquidability.
  Although this abstract semantics substantially simplifies the concrete one,
  it is still infinite-state.
  We establish that the abstract semantics is an over-approximation of the concrete one
  (Theorem~\ref{th:abs-bitml:over-approximation}).
  Further, we show that abstract runs containing only moves of $\pmvA$ alone
  under-approximate concrete runs
  (Theorem~\ref{th:abs-bitml:under-approximation}).

\item We define a finite-state refinement $\finrightarrow{}$ of $\absrightarrow{}$, 
  which allows for model-checking liquidity.
  Liquidity \wrt $\finrightarrow{}$ is proved to be equivalent to liquidity \wrt $\absrightarrow{}$ 
  (Theorem~\ref{th:abs-bitml:fin-liquidity}).

\end{itemize}

\begin{figure}[t]
  \[
  \begin{array}{c}
    \absContr[\pmvA]{\sum_{i \in I} \contrD[i]}
    = \sum_{i \in I} \absContr[\pmvA]{\contrD[i]}
    \qquad
    \absContr[\pmvA]{\withdrawC{B}}
    = \withdrawC{B}
    \\[12pt]
    \absContr[\pmvA]{\putC[\predP]{}{\vec{\secrA}} . \, \contrC}
    = \begin{cases}
      \tau.\,\absContr[\pmvA]{\contrC}
      & \text{if $\predP = \true$ and $\vec{\secrA} \subseteq \secrOf{\pmvA}$}
      \\
      \advPref{\tau.\,\absContr[\pmvA]{\contrC}}
      & \text{otherwise}
    \end{cases}
    \\[20pt]
    \absContr[\pmvA]{\splitC{\big(  \mmid_{i \in I}\; \splitB{\valV[i]}{\contrC[i]} \big)}}
    = 
    \splitC{\big( \mmid_{i \in I}\; \splitB{}{\absContr[\pmvA]{\contrC[i]}} \big)}
    \\[15pt]
    \absContr[\pmvA]{\authC{\pmvB}{\contrD}}
    = \begin{cases}
      \absContr[\pmvA]{\contrD}
      & \text{if $\pmvB = \pmvA$} 
      \\
      \advPref{\absContr[\pmvA]{\contrD}}
      & \text{otherwise}
    \end{cases}
        \qquad
        \absContr[\pmvA]{\afterC{t}{\contrD}}
        = \absContr[\pmvA]{\contrD}
    \\[20pt]
    \absContr[\pmvA]{\adv{\callX{\vec{\sexp}}}}
    = \adv{\callX{}}
    \\[15pt]
    \absConf[{\pmvA,\VarX}]{\confD \mid \confDi} = \absConf[{\pmvA,\VarX}]{\confD} \mid \absConf[{\pmvA,\VarX}]{\confDi}
    \\[12pt]
    \absConf[{\pmvA,\VarX}]{\confContr[\varX]{\contrC}{\valV}} 
    = \confContr[\varX]{\absContr[\pmvA]{\contrC}}{}
    \quad \text{if $x \in \VarX$}
    \qquad\qquad
    \absConf[{\pmvA,\VarX}]{\confD} = \gnil
    \quad
    \text{otherwise}
  \end{array}
  \]
  \caption{Abstraction of contracts and configurations.}
  \label{fig:abs-conf}
\end{figure}

\medskip
We start by defining the abstraction of configurations.
Hereafter, we assume that all the abstractions 
are done \wrt the same honest participant $\pmvA$.

\begin{defi}[Abstraction of configurations]
  \label{def:abs-conf}
  For all sets of contract names $\VarX$,
  we define the contract abstraction function $\absContr[\pmvA]{}$ 
  and the configuration abstraction function $\absConf[{\pmvA,\VarX}]{}$ 
  in~\Cref{fig:abs-conf}.
  We abstract each defining equation
  $\decl{\cVarX}{\vec{\procParamA}} = \contrAdv{\contrG}{\contrC}$ as
  $\cVarX =_{\sharp} \absContr[\pmvA]{\contrC}$.
  The syntax of abstract configurations is implicitly given by the equations
  in Figure~\ref{fig:abs-conf}.
  We further allow recursion variables $\callX{}$ within abstract configurations.
  We identify $\advPref{\advPref{\contrD}}$ with $\advPref{\contrD}$.
\end{defi}

The equations in Figure~\ref{fig:abs-conf} follow the intuition described above.
The prefix $\tau$ used in $\revealname$ is used to abstract the fact that $\pmvA$ alone
can unconditionally reveal some of her secrets.
Instead, we abstract as $\advPref{\tau}$ the case where 
some of the secrets to be revealed do not belong to $\pmvA$, 
or the truth of the predicate $\predP$ is unknown.
Since the secrets have been removed from configurations, 
this is a conservative (but safe) abstraction.

We now describe the abstract semantics of BitML.
In the relation $\absrightarrow{}$ between abstract configurations,
the rules to advertise contracts, for deposits, and for delays are removed.
There is a rule for making a contract $\confContr[x]{\withdrawC{\pmvA}}{\valV}$
reduce to a deposit $\confDep[y]{\pmvA}{\valV}$ is replaced
so that $\confContr[x]{\withdrawC{\pmvA}}{\valV}$
reduces to $\gnil$ (the empty configuration).

\begin{figure}[t]
  \[
  \begin{array}{c}
    \irule
    {}
    {\confContr[\varX]{\withdrawC{\pmvB}}{} \mid \absConfG 
    \absrightarrow{\varX} 
    \absConfG}
    \;\smallnrule{[A-Withdraw]}
    \qquad
    \irule
    {\varY \;\text{fresh}}
    {\confContr[\varX]{\tau.\contrC}{} \mid \absConfG 
    \absrightarrow{\varX} 
    \confContr[\varY]{\contrC}{} \mid \absConfG}
    \;\smallnrule{[A-Rev]}
    \\[20pt]
    \irule
    {\varY[i]\; \text{fresh}}
    {\confContr[\varX]{\splitC{\mmid_i \;\splitB{}{\contrC[i]}}}{} \mid \absConfG
    \absrightarrow{\varX} 
    \mmid_i \;\confContr[{\varY[i]}]{\contrC[i]}{} \mid \absConfG}
    \;\smallnrule{[A-Split]}
    \\[20pt]
    \irule
    {\confContr[\varX]{\contrD}{} \mid \absConfG \absrightarrow{\absLabS} \absConfGi \quad \varX \in \cn{\absLabS}}
    {\confContr[\varX]{\contrD + \contrC}{} \mid \absConfG \absrightarrow{\absLabS} \absConfGi}
    \;\smallnrule{[A-Branch]}
    \quad
    \irule
    {\confContr[\varX]{\contrD}{} \mid \absConfG \absrightarrow{\absLabS} \absConfGi \quad \varX \in \cn{\absLabS}}
    {\confContr[\varX]{\advPref{\contrD}}{} \mid \absConfG \absrightarrow{\advPref{\absLabS}} \absConfGi}
    \;\smallnrule{[A-Ext]}
    \\[25pt]
    \irule
    {\cVarX =_{\sharp} \contrC \quad \varY \;\text{fresh}}
    {\confContr[\varX]{\adv{\callX{}}}{} \mid \absConfG \absrightarrow{\advPref{\varX}} \confContr[\varY]{\contrC}{} \mid \callX{} \mid \absConfG}
    \;\smallnrule{[A-Rngt]}
  \end{array}
  \]
  \caption{Abstract BitML semantics.}
  \label{fig:abs-bitml}
\end{figure}

\begin{defi}[Abstract semantics]
  \label{def:abs-bitml}
  We define the relation $\absrightarrow{}$ 
  between abstract configurations in~\Cref{fig:abs-bitml}.
  We use $\absLabS$ to range over labels, which have the form $\varX$ or 
  the form $\advPref{\varX}$, where $\varX$ is a contract name.
  We identify $\advPref{\advPref{\absLabS}}$ with $\advPref{\absLabS}$.
  We define $\mathcal{L}^{\sharp}$ as the subset of labels of the form $\varX$.
  An \emph{abstract run} $\absRunS$ is a sequence 
  \(
  \confG[0]
  \absrightarrow{}
  \confG[1]
  \absrightarrow{}
  \cdots
  \).
\end{defi}

We briefly comment the rules in Figure~\ref{fig:abs-bitml}, which define
the abstract semantics $\absrightarrow{}$ of BitML.
Most rules are straightforward.
In rule \nrule{[A-Ext]}, we record in the label $\advPref{\absLabS}$
the fact that reducing a contract $\advPref{\contrD}$ might require the cooperation
from some other participant.
In rule \nrule{[A-Rngt]}, we record through the term $\callX{}$ in the configuration 
the fact that we have unfolded the recursion variable $\callX{}$.
This is not strictly needed to make the abstract semantics safely approximate 
the concrete one: rather, it is a technical expedient to build the finite-state semantics
$\finrightarrow{}$ on-top of $\absrightarrow{}$.

\paragraph{Correspondence between the semantics.}

We now establish
a correspondence between the abstract and the concrete semantics of BitML.
In all the statements below, we assume that the concrete configuration
$\confG$ is reachable from an initial configuration.

We introduce below the notion of \emph{descendants} of a set of contracts, 
which is dual to that of originator in Definition~\ref{def:origin}.
This notion is exploited whenever we reduce an abstract configuration
$\absConf[{\pmvA,\VarX}]{\confG}$ mimicking the steps of a concrete run $\runS$
from $\confG$ to $\confGi$.
Intuitively, the resulting configuration will be an abstraction of $\confGi$
against the descendants of $\VarX$.

\begin{defi}[Descendants]
  \label{def:desc}
  For all concrete configurations $\confG$, 
  runs $\runS = \confG \concrightarrow{}^* \confGi$,
  and set of contract names $\VarX$,
  we define the set of contract names $\descen[\confG]{\runS}{\VarX}$
  as follows:
  \[
  \descen[\confG]{\runS}{\VarX} 
  \; = \;
  \setcomp{\varX \in \cn{\confGi}}{\origin[\confG]{\runS}{x} \in \VarX}
  \]
\end{defi}

Theorem~\ref{th:abs-bitml:over-approximation}
establishes that each concrete run $\runS = \confG \concrightarrow{}^* \confGi$
has a corresponding abstract run,
whose first element is the abstraction of $\confG$,
and the last element is the abstraction of $\confGi$, plus recursion variables.

\begin{thm}[Over-approximation]
  \label{th:abs-bitml:over-approximation}
  Let $\VarX \subseteq \cn{\confG}$, 
  and let $\runS = \confG \concrightarrow{}^* \confGi$ be a concrete run.
  Then, there exist $\callX{}_1, \callX{}_2, \ldots$ such that:
  \[
  \absConf[{\pmvA,\VarX}]{\confG}
  \absrightarrow{}^* 
  \absConf[{\pmvA,\descen[\confG]{\runS}{\VarX}}]{\confGi} 
  \mid \mmid_i \, \callX{}_i
  \]
\end{thm}
\begin{proof}\emph{(sketch)}
  By induction on the number of concrete moves. 
  We check that each move $\confD \concrightarrow{} \confDi$ 
  is matched by zero or one abstract moves.
  First, note that the abstract configuration contains only contracts in $\VarX$.
  We have the following cases, according to the concrete move:
  \begin{itemize}
    
  \item If $\confD \concrightarrow{} \confDi$ does neither consume nor create a contract, then
    the abstractions of $\confD$ and $\confDi$ are the same, 
    so we match the concrete move with zero abstract moves.
    
  \item If $\confD \concrightarrow{} \confDi$ creates contracts without consuming any contract, 
    then it is an instance of rule \nrule{[C-Init]}. 
    In this case, the created contract is not a descendant of any contract in $\VarX$,
    so it does not occur in the abstraction of $\confDi$.
    As in the previous case, we match the concrete move with zero abstract moves.
    
  \item If $\confD \concrightarrow{} \confDi$ consumes a contract $\varX \not \in \VarX$, then the abstraction
    discards the contract $\varX$ and its descendants, and so we match the concrete move with zero abstract moves.
    
  \item If $\confD \concrightarrow{} \confDi$ consumes a contract $\varX \in \VarX$, then one of the rules
    \nrule{[C-Split]}, \nrule{[C-Rev]}, \nrule{[C-Withdraw]} or \nrule{[C-Rngt]} have been used. 
    The moves \nrule{[C-Split]}, \nrule{[C-Withdraw]} and \nrule{[C-Rngt]}
    are matched by the corresponding \nrule{[A-*]} abstract moves.
    The move \nrule{[C-Rev]} is matched either by \nrule{[A-Rev]}
    or \nrule{[A-Ext]}, depending on how the $\revealname$ was abstracted.
    Note that \nrule{[C-Rngt]} may instantiate the formal parameters 
    in the contract with actual values.
    However, these values are abstracted away by $\absContr[\pmvA]{}$,
    so this instantiation is immaterial.
    The term $\callX{}$ introduced by \nrule{[A-Rngt]} is accounted for
    by the statement (it will belong to $\|_i\, \callX{}_i$).
    \qedhere
  \end{itemize}
\end{proof}

Theorem~\ref{th:abs-bitml:under-approximation} associates some abstract runs with concrete runs. 
More specifically, it considers an abstract run starting from the abstraction of a concrete configuration $\confG$ (plus recursion variables),
and whose labels represent actions performable by $\pmvA$ 
(\eg, the actions in $\mathcal{L}^{\sharp}$).
The theorem constructs a concrete run starting from $\confG$, 
whose last configuration, once abstracted, matches the last configuration in the abstract run.
Further, the labels in the concrete run belong to 
the set of concrete actions performable by $\pmvA$ 
(which we denote by $\mathcal{L}^{\flat}_{\pmvA}$, 
\ie all the concrete labels except those of the form $\pmvB:\ell$ with $\pmvB \neq \pmvA$).

\begin{thm}[Under-approximation]
  \label{th:abs-bitml:under-approximation}
  Let $\VarX \subseteq \cn{\confG}$, and assume that:
  \begin{align*}
    & \absConf[{\pmvA,\VarX}]{\confG}
      \mid \! \mmid_i \, \callX{}_i
      \; \absrightarrow{}^* \;
      \confGi[\sharp]
    && \text{with labels in $\mathcal{L}^{\sharp}$}
    \\
    \intertext{Then, there exist $\runS$ and $\confGi$ such that:}
    & \runS = \confG \; \concrightarrow{}^* \; \confGi
    && \text{with labels in $\mathcal{L}^{\flat}_{\pmvA}$, and $\confGi[\sharp] = \absConf[{\pmvA,\descen[\confG]{\runS}{\VarX}}]{\confGi} \mid \! \mmid_i \, \callX{}_i$}
  \end{align*}
\end{thm}
\begin{proof}\emph{(sketch)}
  By induction on the number of moves. 
  Since each abstract move must be in $\mathcal{L}^{\sharp}$, 
  it must be derived by one of the rules 
  \nrule{[A-Withdraw]}, \nrule{[A-Rev]}, or \nrule{[A-Split]}, 
  possibly as a premise of \nrule{[A-Branch]}. 
  To fall into these cases, 
  the concrete contract must have one of the following forms:
  \begin{itemize}

  \item $\withdrawC{\pmvB}$ or $\splitC{\cdots}$. 
    Here, the abstract move is matched by the concrete moves \nrule{[C-Withdraw]} or 
    \nrule{[C-Split]}, the labels of which belong to $\mathcal{L}^{\flat}_{\pmvA}$.

  \item $\putC[\true]{}{\vec{\secrA}}$ where $\vec{\secrA}$ are secrets of $\pmvA$, the participant \wrt whom we abstract. 
    In this case, the abstract move \nrule{[A-Rev]} is matched by \nrule{[C-Rev]}, the label of which belongs to $\mathcal{L}^{\flat}_{\pmvA}$. 
    If the secrets have not been revealed, yet, we perform one or more \nrule{[C-AuthRev]} moves (the labels of which also belong to $\mathcal{L}^{\flat}_{\pmvA}$) before the \nrule{[C-Rev]} move.
    Note that the premise $N \neq \bot$ in rule \nrule{[C-AuthRev]} is satisfied, because we are assuming that the participant $\pmvA$ against whom we are abstracting is honest ($\pmvA \in \PartT$), and rule $\nrule{[C-AuthCommit]}$ ensures that honest participants have committed
    to secrets $N \neq \bot$.

  \item any of the above, constrained by $\afterC{t}{\cdots}$ or $\authC{\pmvB}{\cdots}$ with 
    $\pmvB = \pmvA$. Here, if needed, we perform a \nrule{[C-Delay]} to advance the current time until time $t$ (a move in $\mathcal{L}^{\flat}_{\pmvA}$). 
    Then, if needed, $\pmvA$ can perform a \nrule{[C-AuthBranch]} to authorize the branch 
    (this move is also in $\mathcal{L}^{\flat}_{\pmvA}$). 
    After that, the constraints are satisfied, so we proceed as in the previous items.
    \qedhere
  \end{itemize}
\end{proof}

\paragraph{Safe approximation of liquidity.}

Theorem~\ref{th:abs-bitml:liquidity} below
establishes the soundness of our abstraction \wrt liquidity.
It exploits the following corollary for liquidability.

\begin{cor}[Abstraction soundness against liquidability]
  \label{cor:abs-bitml:liquidability}
  Let $\varX \in \VarX \subseteq \cn{\confG}$.
  If $\varX$ is liquidable in 
  $\absConf[{\pmvA,\VarX}]{\confG} \mid \vec{\callY{}}$ 
  \wrt $\absrightarrow{} / \mathcal{L}^{\sharp}$,
  then
  $\varX$ is liquidable in $\confG$ 
  \wrt $\concrightarrow{} / \mathcal{L}^{\flat}_{\pmvA}$.
\end{cor}
\begin{proof}
  Since $\varX$ is liquidable in 
  $\absConf[{\pmvA,\VarX}]{\confG} \mid \vec{\callY{}}$,
  there exists an abstract run with labels in $\mathcal{L}^{\sharp}$
  and leading to a configuration $\absConfGi$
  without contracts originated from $\varX$.
  By Theorem~\ref{th:abs-bitml:under-approximation},
  there is a corresponding concrete run with labels in $\mathcal{L}^{\flat}_{\pmvA}$
  starting from $\confG$ and leading to a configuration $\confGi$
  the abstraction of which is $\absConfGi$.
  Since the abstraction preserves all the contracts originated from $\varX$,
  in $\confGi$ there are none.
  Therefore, $\varX$ is liquidable in $\confG$.
\end{proof}

\begin{thm}[Abstraction soundness against liquidity]
  \label{th:abs-bitml:liquidity}
  Let $\varX \in \VarX \subseteq \cn{\confG}$.
  If $\varX$ is liquid in $\absConf[{\pmvA,\VarX}]{\confG}$ 
  \wrt $\absrightarrow{} / \mathcal{L}^{\sharp}$,
  then $\varX$ is liquid in $\confG$ 
  \wrt $\concrightarrow{} / \mathcal{L}^{\flat}_{\pmvA}$.
\end{thm}
\begin{proof}
  Let $\runS  = \confG \xrightarrow{} \cdots \xrightarrow{} \confGi$ 
  be a concrete run, and let $\varY$ be such that 
  $\origin[\confG]{\runS}{\varY} = \varX$. 
  By Theorem~\ref{th:abs-bitml:over-approximation}, 
  there exists a corresponding abstract run:
  \[
  \absConfG \; = \; \absConf[{\pmvA,\VarX}]{\confG}
  \; \absrightarrow{}^* \; 
  \absConfGi \; = \; \absConf[{\pmvA,\descen[\confG]{\runS}{\VarX}}]{\confGi} 
  \mid \mmid_i \, \callX{}_i
  \]
  Since $\varX$ is liquid in $\absConfG$ and $\varY$ is a descendant of $\varX$
  in the abstract run, 
  $\varY$ is liquidable in $\absConfGi$.
  By Corollary~\ref{cor:abs-bitml:liquidability},
  $\varY$ is liquidable in $\confGi$.
  Therefore, $\varX$ is liquid in $\confG$.
\end{proof}

\paragraph{A finite-state abstraction of BitML}

The abstract semantics $\absrightarrow{}$ is infinite-state,
even up-to renaming of contract names.
Indeed, each application of rule \nrule{[A-Rngt]} unfolds a contract $\callX{}$,
which can spawn other parallel contracts before recursing.
Therefore, the number of parallel contracts in reachable abstract configurations 
can grow unboundedly.
This hinders verification based on model-checking the whole state space.

We now introduce another abstract semantics, called $\finrightarrow{}$,
which has a \emph{finite} number of reachable configurations
(up-to renaming of contract names),
and which preserves the liquidity \wrt $\absrightarrow{}$:
more specifically, a contract is liquid \wrt $\absrightarrow{}$ 
if and only if it is liquid \wrt $\finrightarrow{}$.
We define the transition relation $\finrightarrow{}$ by the following rule:
\[
\irule
{\confG \absrightarrow{\absLabS} \confGi
  \quad  
  \nexists \confD, \callX{} : \confGi = \confD \mid \callX{} \mid \callX{}}
{\confG \finrightarrow{\absLabS} \confGi}
\]

Intuitively, $\finrightarrow{}$ can mimic any move of $\absrightarrow{}$,
except for the moves \nrule{[A-Rngt]} which renegotiate a contract 
which has already been renegotiated.
Technically, this is ensured by constraining the configuration to contain at most one 
occurrence of each recursion variable~$\callX{}$.

\begin{thm}
  \label{th:liquidity-fin-bitml:finite-state}
  Starting from any abstract configuration,
  the set of states reachable through $\finrightarrow{}$ is finite,
  up-to renaming of contract names.
\end{thm}
\begin{proof}
  Direct consequence of the definition of the abstract semantics. 
\end{proof}

The following lemma establishes the equivalence of $\absrightarrow{}$ and
$\finrightarrow{}$ against liquidability,
and it is instrumental to prove that they also agree on liquidity
(Theorem~\ref{th:abs-bitml:fin-liquidity}).

\begin{lem}
  \label{lem:abs-bitml:fin-liquidability}
  Let $\absConfG$ be an abstract configuration
  without multiple occurrences of any $\callX{}$.
  Then, for all contract names $\varX$:
  \[
  \varX \text{ liquidable in } \absConfG 
  \text{ \wrt } \absrightarrow{} / \mathcal{L}^{\sharp}
  \; \iff \;
  \varX \text{ liquidable in } \absConfG
  \text{ \wrt } \finrightarrow{} / \mathcal{L}^{\sharp}
  \]
\end{lem}
\begin{proof}
  Recall that liquidability only considers moves in $\mathcal{L}^{\sharp}$.
  The only difference between $\absrightarrow{}$ and $\finrightarrow{}$ is 
  that the latter forbids, in some configurations, 
  the application of rule \nrule{[A-Rngt]}.
  Since the label of this rule is not in $\mathcal{L}^{\sharp}$, this
  difference is immaterial for liquidability.
\end{proof}

\begin{thm}
  \label{th:abs-bitml:fin-liquidity}
  Let $\absConfG$ be an abstract configuration
  without occurrences of any $\callX{}$.
  Then, for all contract names $\varX$:
  \[
  \varX \text{ liquid in } \absConfG 
  \text{ \wrt } \absrightarrow{} / \mathcal{L}^{\sharp}
  \; \iff \;
  \varX \text{ liquid in } \absConfG
  \text{ \wrt } \finrightarrow{} / \mathcal{L}^{\sharp}
  \]
\end{thm}
\begin{proof}
  \emph{(sketch)}
  For $\Rightarrow$, consider a run $\absRunS$ of $\finrightarrow{}$ 
  starting from $\absConfG$ and leading to some $\absConfGi$, and let $\varY$ be a descendant of $\varX$.
  Since $\absrightarrow{}$ includes $\finrightarrow{}$, 
  $\absRunS$ is also a run of $\absrightarrow{}$.
  Since $\varX$ is liquid in $\absConfG$ \wrt $\absrightarrow{}$, then
  $\varY$ must be liquidable in $\absConfGi$ \wrt $\absrightarrow{}$.
  By Lemma~\ref{lem:abs-bitml:fin-liquidability}, it follows that
  $\varY$ is liquidable also in $\finrightarrow{}$.
  Therefore, $\varX$ is liquid in $\absConfG$ \wrt $\finrightarrow{}$.

  For $\Leftarrow$, take any contract $\confContr[\varY]{\contrC}{}$ which originates from the contract $\varX$ in a run of $\absrightarrow{}$ starting from $\absConfG$. 
  By contradiction, assume that $\varY$ is not liquidable for $\absrightarrow{}$.
  In the $\absrightarrow{}$ run, $\confContr[\varY]{\contrC}{}$ might be reached after several renegotiations involving some defined contracts $\callX{}_1,\callX{}_2,\ldots$ more than once. 
  In such case, it is also possible to reach $\contrC$ 
  without renegotiating the same $\callX{}_i$ more than once, 
  because each use of rule \nrule{[A-Rngt]} on $\callX{}_i$ spawns the same contract, and contracts in parallel do not interact.
  Hence, there exists some $\absrightarrow{}$ run leading to a non liquidable 
  $\confContr[\varZ]{\contrC}{}$ and whose configurations never include 
  any $\callX{}_i$ more than once.
  This run is therefore also a $\finrightarrow{}$ run, 
  leading to a non liquidable $\varZ$ according to Lemma~\ref{lem:abs-bitml:fin-liquidability}
  --- contradiction with the liquidity of $\varX$ \wrt $\finrightarrow{}$.
\end{proof}

\section{The BitML toolchain}
\label{sec:toolchain}

\begin{figure}
  \footnotesize
    \begin{tikzpicture}[>=triangle 45]
      \node[block] (a) {BitML \\ on DrRacket};
      \node[inputblock, below =0.7cm of a]   (ap){Queries};
      \node[inputblock, above =0.7cm of a]   (ac){Contract};

      \node[block, above right = 0.2cm and 2cm of a] (b) {Abstract BitML \\ semantics};
      \node[block, right =0.7cm of b]   (mc){Model \\ checker};
      \node[outputblock,right =1cm of mc]   (qr){Query \\ result};

      \node[block, below =2cm of b]   (c1){BitML to \\ Balzac};
      \node[block, right =0.7cm of c1]   (c2){Balzac to \\ Bitcoin};
      \node[outputblock,right =1cm of c2]   (tx){Bitcoin \\ transactions};

      \node at (4.2,3) {\textsc{Security Analyzer}};
      \draw [ultra thick, draw=black, fill=none, opacity=0.05]
      (2.8,0.5) -- (2.8,2.8) -- (9.4,2.8) -- (9.4,0.5) -- cycle;

      \node at (3.5,-0.5) {\textsc{Compiler}};
      \draw [ultra thick, draw=black, fill=none, opacity=0.05]
      (2.8,-2.9) -- (2.8,-0.7) -- (9.4,-0.7) -- (9.4,-2.9) -- cycle;

      \draw[arrows=->,line width=0.5pt] (ac.south) -| (a.north);
      \draw[arrows=->,line width=0.5pt] (ap.north) -| (a.south);

      \draw[arrows=->,line width=0.5pt] (a.east) -| (2.25,0) |- (b.west);
      \draw[arrows=->,line width=0.5pt] (a.east) -| (2.25,0) |- (c1.west);
      \draw[arrows=->,line width=0.5pt] (b.east) |- (mc.west);
      \draw[arrows=->,line width=0.5pt] (mc.east) |- (qr.west);

      \draw[arrows=->,line width=0.5pt] (c1.east) |- (c2.west);
      \draw[arrows=->,line width=0.5pt] (c2.east) |- (tx.west);
    \end{tikzpicture}
  \caption{Toolchain architecture.}
  \label{fig:toolchain-architecture}
\end{figure}

We now describe the BitML toolchain,
whose architecture is displayed in~\Cref{fig:toolchain-architecture}.
The development workflow is the following:
\begin{inlinelist}[(a)]
\item write the BitML contract in the DrRacket IDE;
\item verify that the contract is liquid \wrt the given participant;
\item compile the contract to Bitcoin transactions;
\item execute the contract, by appending these transactions to the Bitcoin blockchain
  according to the chosen strategy.
\end{inlinelist}
The verifier implements the abstract BitML semantics in Maude, a
model-checking framework based on rewriting logic~\cite{maude01}.
The toolchain is open-source%
\footnote{\url{https://github.com/bitml-lang}};
a tutorial is available online%
\footnote{\url{https://blockchain.unica.it/bitml}},
including references to our experiments on the Bitcoin testnet.

\paragraph{Benchmarks}

To evaluate our toolchain, 
we use a benchmark of representative use cases%
\footnote{\url{https://github.com/bitml-lang/bitml-compiler/tree/master/examples/benchmarks}}
including financial contracts~\cite{Thompson18isola,Biryukov17wtsc}, 
auctions, lotteries~\cite{Andrychowicz16cacm,Miller16zerocollateral} 
and gambling games.
For each contract in the benchmark, 
we display in \Cref{fig:evaluation:benchmarks}
the number of involved participants,
the number of transactions obtained by the compiler,
and the number of Maude rewrites for checking liquidity.
Notably, the tool automatically verifies that all the contracts
in the benchmark are liquid.
The verification time for all the benchmarks is in the order of milliseconds
on a consumer-grade laptop, 
except for $\ZCBii$, which requires $\sim 1 s$.

\begin{table}[t!]
  \centering
  \small
  \begin{tabular}{cccc}\toprule
    Contract & Participants & Transactions & Rewrites \\
    \midrule
    Mutual timed commitment & 2 & 15 & 72 \\
    Mutual timed commitment & 3 & 34 & 207 \\
    Mutual timed commitment & 4 & 75 & 644 \\
    Mutual timed commitment & 5 & 164 & 2097 \\
    Escrow (early fees) & 3 & 12 & 104 \\    
    Escrow (late fees) & 3 & 11 & 110 \\ 
    Zero Coupon Bond & 3 & 8 & 189 \\
    Coupon Bond & 3 & 18 & 9101 \\
    Future$(C)$ & 3 & 5 + $\mathit{T}_C$ & 136 + $R_C$ \\
    Option$(C,D)$ & 3 & 14 + $T_{C} + T_{D}$ & 162 + $R_C + R_D$ \\
    Lottery (quadratic collateral) & 2 & 15 & 1466 \\
    Lottery ($0$ collateral) & 2 & 8 & 31 \\
    Lottery ($0$ collateral) & 4 & 587 & 167 \\
    Rock-Paper-Scissors & 2 & 23 & 2322 \\
    Morra game & 2 & 40 & 89 \\
    Shell game & 2 & 23 & 48 \\
    Auction (2 turns) & 2 & 42 & 218 \\
    Coin flipping game & 2 & 32 & 563 \\
    Zero coupon bond (v3) & 3 & 44 & 1196813 \\
    \bottomrule
  \end{tabular}
  \caption{Benchmarks for the BitML toolchain.}
  \label{fig:evaluation:benchmarks}
\end{table}

We compare the performance of our tool against~\cite{Andrychowicz14formats}, 
which models Bitcoin contracts in Uppaal,
a model-checking framework based on Timed Automata.
The most complex contract modelled in~\cite{Andrychowicz14formats}
is the mutual timed commitment with 2 participants:
this requires \mbox{$\sim 30$s} to be verified in Uppaal,
while our tool verifies the same property in \mbox{$\sim 1$ms}.
This speedup is due to the higher abstraction level of BitML
over~\cite{Andrychowicz14formats}, 
which operates at the (lower) level of Bitcoin transactions.
Note that increasing the number of participants in the mutual timed commitment significantly affects verification time. This is because the size of the contract increases and there are more committed secrets in play. 
Indeed, a larger number of secrets increases the branching in the (abstract) BitML LTS, since in each state any secret not revealed so far can be revealed. 
This increases the size of the state space, and consequently the complexity of model checking.

We observe that the number of Maude rewritings for $\ZCBii$ is particularly large. 
This is due to the contract generating many parallel components, which cause the explosion
of the state space.
More specifically, $\ZCBii$ performs 10 steps of recursion, each one involving a split, which creates several parallel sub-contracts (half of them with a choice).

The exponential blow-up due to parallel components is a common performance issue of general model checking, and so it also occurs in our setting.
Indeed, the benchmarks in Table 1 only include the contracts for which we have managed to verify liquidity. 
We have not explicitly looked for the simplest examples 
for which the verification is practically unfeasible. 
However, by the discussion above, we expect that increasing 
the number of recursion steps by $\sim$20 units in the $\ZCBii$ 
contract would be enough to make the model checker exhaust the available 
resources.
Besides increasing the number of recursion steps, 
there are many other ways to design a contract for which verifying liquidity is unfeasible.
For instance, if the contract contains a split of $N$ withdraw actions, 
the size of the state space is at least $O(2^N)$.
A similar size is obtained for a contract with $N$ reveal actions in sequence, 
since at each point of the computation one has the option of revealing 
any of their secrets.
The very same problem is witnessed by a contract requiring $N$ authorizations, \eg as in
$\authC{{\pmvA[1]}}{\cdots : \authC{{\pmvA[N]}}{\contrD}}$.

\paragraph{Limitations}

One of the main difficulties that we have encountered 
in developing contracts is that some complex BitML specifications
can not be compiled to pure Bitcoin,
because of the 520-byte limit on the size of each value pushed to the 
evaluation stack~\cite{btc520bytes}.
In some cases, we managed to massage the BitML contract
so to make its compilation respect the constraint.
For instance, a pattern  that easily violates the constraint
is the following:

\begin{center}
\scalebox{1.18}{
\begin{lstlisting}[language=bitml,classoffset=1,morekeywords={},classoffset=2,morekeywords={FundsA,Commit,Reveal},framexbottommargin=0pt,framextopmargin=0pt]
(choice (revealif (b) (pred (p0)) (C0))
        (revealif (b) (pred (p1)) (C1))
        (after T      (C2)))
\end{lstlisting}
}
\end{center}

The \code{choice} is compiled
into a transaction whose redeem script encodes
the disjunction of \emph{three} logical conditions,
corresponding to the three branches of the \code{choice}.
Depending on the predicates \code{p0} and \code{p1},
and on the number of participants in the contract,
this script may violate the 520-byte constraint.
A workaround is to rewrite the pattern above into:

\begin{center}
\scalebox{1.18}{
\begin{lstlisting}[language=bitml,classoffset=1,morekeywords={},classoffset=2,morekeywords={FundsA,Commit,Reveal},framexbottommargin=0pt,framextopmargin=0pt]
(choice (revealif (b) (pred (p0)) (C0))
        (after T (tau (choice
                      (revealif (b) (pred (p1)) (C1))
                      (after T1     (C2))))))
\end{lstlisting}
}
\end{center}

In this case the compilation includes two transactions, 
corresponding to the two \code{choice}s.
The scripts of these transactions encode the disjunction of \emph{two} 
logical conditions, corresponding to the two branches of the \code{choice}s.
Using this workaround we have managed to compile a 4-players lottery
into standard transactions, 
at the price of increasing the number of transactions 
(587 for the standard version \emph{vs.}~138 for the nonstandard one).
Similar techniques (\eg simplification of predicates%
\footnote{\url{https://github.com/bitml-lang/bitml-compiler/blob/master/bitml/exp.rkt}}
) allowed us to compile 
all the contracts in~\Cref{fig:evaluation:benchmarks}
into standard Bitcoin transactions.

In general, the 520-byte constraint intrinsically 
limits the expressiveness of Bitcoin contracts:
for instance, since public keys are 33 bytes long, 
a contract which needs to simultaneously verify 15 signatures
can not be implemented using standard transactions.

\section{Variants of liquidity}
\label{sec:variants-liquidity}

We now discuss some variants of the notion of liquidity of Section~\ref{sec:liquidity}.

\paragraph{Liquidity under a strategy}

The runs $\runS = \confG \xrightarrow{} \cdots \xrightarrow{} \confGi$
in Definition~\ref{def:liquid} allow any participant to perform
\emph{any} enabled move.
For instance, consider the contract:
\[
\contrAdv
{\persdep{\pmvA}{1}{\varX} \mid \presecret{\pmvA}{\secrA} \mid \presecret{\pmvA}{\secrAi}}
{
  \big(
  \putC{}{\secrA} . \, \withdrawC{\pmvA}
  \, + \,
  \putC{}{\secrAi}. \, \authC{\pmvB}{\withdrawC{\pmvA}}
  \big)
}
\]
This contract is \emph{not} liquid for $\pmvA$: 
indeed, if $\pmvA$ performs $\putC{}{\secrAi}$, the reduct 
$\authC{\pmvB}{\withdrawC{\pmvA}}$ is not liquidable by $\pmvA$ alone,
since $\pmvB$ may refuse to give his authorization.
To overcome this issue, $\pmvA$ can follow the strategy of
always performing the $\putC{}{\secrA}$. 
In this way, she is sure to be able liquidate the contract
from any reachable state.
A possible extension of the notion of liquidity 
in Definition~\ref{def:liquid} is to make it parametric on 
$\pmvA$'s strategy, and consider only the runs
$\runS = \confG \xrightarrow{} \cdots \xrightarrow{} \confGi$
which are coherent with it.

\paragraph{Multiparty liquidity}

Definition~\ref{def:liquid} requires that $\pmvA$ \emph{alone} 
can liquidate each descendent of the contract.
We can relax it by considering a set of collaborative participants.
For instance, consider an escrow contract between $\pmvA$ and $\pmvB$,
involving also a mediator $\pmv{M}$:
\begin{align*}
  & \Escrow
  =
    \authC{\pmvA}{\withdrawC{\pmvB}}
    \; + \;
    \authC{\pmvB}{\withdrawC{\pmvA}}
    \; + \; 
    \authC{\pmvA}{\Resolve} 
    \; + \;
    \authC{\pmvB}{\Resolve}
  \\[2pt]
  & \Resolve
  =
  \splitname (
    \splitB{0.1}{\withdrawC{\pmvM}} 
  \\
  & \hspace{72pt}
    \mid
   \; \splitB{0.9}{\authC{\pmvM}{\withdrawC{\pmvA}} + \authC{\pmvM}{\withdrawC{\pmvB}}}
    )
\end{align*}

After the contract has been stipulated, 
$\pmvA$ can choose to pay $\pmvB$, by authorizing the first branch. 
Similarly, $\pmvB$ can allow $\pmvA$ to take her money back,
by authorizing the second branch. 
If they do not agree, any of them can invoke a mediator $\pmvM$ to resolve the dispute,
invoking a $\Resolve$ branch.
There, the initial deposit (say, of $\valV\BTC$) is split in two parts:
$0.1\valV$ goes to the mediator,
while $0.9\valV$ is assigned either to $\pmvA$ and $\pmvB$,
depending on $\pmvM$'s  choice.
This contract is \emph{not} liquid for $\pmvA$,
because $\pmvB$ can invoke the mediator, who can refuse to act, 
freezing the funds within the contract (similarly for $\pmvB$).
Instead, assuming that $\pmv{M}$ is collaborative, 
the contract is liquid for both $\pmvA$ and $\pmvB$.
Indeed, a collaborative $\pmv{M}$ will always authorize either the 
$\withdrawC{\pmvA}$ or the $\withdrawC{\pmvB}$
to unlock $0.9\valV$.
Multiparty liquidity where all participants are collaborative 
was used \eg in~\cite{Tsankov18ccs} in the context of Ethereum contracts.

\paragraph{Quantitative liquidity}

Definition~\ref{def:liquid} requires that no funds remain frozen within the contract.
However, in some cases $\pmvA$ could accept the fact that
a portion of the funds remain frozen, especially when these funds
would be assigned to other participants.
We could define a contract \emph{$\valV$-liquid} for $\pmvA$
if at least $\valV$ bitcoins are guaranteed to be redeemable
by anyone.
For instance, $\Lottery[\Win]$ of Section~\ref{sec:examples}
is non-liquid for $\pmvA$, but it is $4\BTC$-liquid.
Instead, $\Lottery[{\Win[2]}]$ is $6\BTC$-liquid, and then also liquid,
under this strategy.
A refinement of this notion could require that at least $\valV\BTC$
are transferred to $\pmvA$, rather than to any participant.
Under this notion, both $\Lottery[\Win]$ and $\Lottery[{\Win[2]}]$ 
would be $2\BTC$-liquid for $\pmvA$.

Other variants of liquidity may take into account the time when funds become liquid,
the payoff of strategies (\eg, ruling out irrational adversaries),
or fairness issues.

\section{Variants of contract primitives}
\label{sec:variants-bitml}

The renegotiation primitive we have proposed for BitML 
is motivated by its simplicity, and by the possibility of compiling 
into standard Bitcoin transactions. 
By adding some degree of complexity, we can devise more general primitives, 
which could be useful in certain scenarios.
We discuss below some alternatives.

\paragraph{Renegotiation-time parameters.}

The primitive $\adv{\callX{\vec{\sexp}}}$ 
allows participants to choose at run-time only the deposit variables 
used in the renegotiated contracts,
and to commit to new secrets.
A possible extension is to allow participants
to choose at run-time \emph{arbitrary} values for the renegotiation 
parameters $\vec{\sexp}$.

For instance, consider a mortgage payment, 
where a buyer $\pmvA$ must pay $10\BTC$ to a bank $\pmvB$ in 10 installments.
After $\pmvA$ has paid the first five installments (of $1\BTC$ each), 
the bank might propose to renegotiate the contract, 
varying the amount of the installment.
Using the BitML renegotiation primitive presented in Section~\ref{sec:bitml}, 
we could not model this contract, since the new amount and the number 
of installments are unknown at the time of the original stipulation. 
Technically, the issue is that the primitive $\adv{\callX{\vec{\sexp}}}$
only involves static expressions $\sexp$, 
the value of which is determined at stipulation time. 

\newcommand{\IPP}[2][]{\ensuremath{\cVar{IPP}{\langle {#2} \rangle}_{#1}}\xspace}

To cope with non-statically known values, 
we could extend guarded contracts with terms of the form 
$\adv{\callX{\pmvB : \valV}}$, 
declaring that the value $\valV$ is to be chosen by $\pmvB$ 
at renegotiation time.
For instance, this would allow to model our installments payment plan
as $\IPP{1}$, with the following defining equations:
\begin{align*}
  \IPP{\procParamA<5}
  & = \contrAdv{\persdep{\pmvA}{1}{\varphX}}
    {\big(
    \splitC{\splitB{1}{\withdrawC{\pmvB}} \mid \splitB{0}{\adv{\IPP{\procParamA+1}}}}
    \big)}
  \\
  \IPP{5}
  & = \contrAdv{\persdep{\pmvA}{1}{\varphX}}
    {\big(
    \splitC{\splitB{1}{\withdrawC{\pmvB}} \mid \splitB{0}{\adv{\callY{\pmvB:k,\pmvB:v}}}}
    \big)}
  \\
  \callY{\procParamA \neq 1,\procParamB}
  & = \contrAdv{\persdep{\pmvA}{\procParamB}{\varphX}}
    {\big(
    \splitC{\splitB{\procParamB}{\withdrawC{\pmvB}} \mid \splitB{0}{\adv{\callY{\procParamA-1,\procParamB}}}}
    \big)}
  \\
  \callY{1,\procParamB}
  & = \contrAdv{\persdep{\pmvA}{\procParamB}{\varphX}}
    {
    \;\withdrawC{\pmvB}
    }
\end{align*}
where in $\IPP{5}$, the bank chooses the number of installments $k$,
as well as the amount $v$ of each installment.
Note that if $\pmvA$ does not agree with these values,
the renegotiation fails. 
A more refined version of the contract should take this possibility into 
account, by adding suitable compensation branches.
Although adding the new primitive would moderately increase the complexity 
of the semantics and of the compiler, 
this extension can still be implemented on top of standard Bitcoin.

\paragraph{Renegotiation with a given set of participants.}

As we have remarked in Section~\ref{sec:bitml},
a renegotiation can be performed only if \emph{all} 
the participants of the contract agree.
We could relax this, by just requiring the agreement of 
a \emph{given} set of participants 
(possibly, not among those who originally stipulated the contract).

For instance, consider an escrow service between a buyer $\pmvA$ 
and a seller $\pmvB$ for the purchase of an item worth $1\BTC$.
The normal case is when the $\pmvA$ authorizes the transfer of $1\BTC$ 
after receiving the item, but it may happen that a dishonest $\pmvB$ 
never ships the item, or that a dishonest buyer never authorizes the payment.
To cope with these cases, the participants could renegotiate the contract, 
including an escrow service $\pmvM$ which mediates the dispute:
\begin{align*}
  & \authC{\pmvA}{\withdrawC{\pmvB}} \, + \, \authC{\pmvB}{\withdrawC{\pmvA}} 
   + \; \extadv{\pmvA:\pmvM}{\call{\Refund[\pmvA]}{}} 
    \, + \, \extadv{\pmvB:\pmvM}{\call{\Refund[\pmvB]}{}}
  \\[2pt]
  & \Refund[\pmvP]
  = \contrAdv{\persdep{\pmvP}{0.1}{\varphX}}
    {\;\splitC \big(
    \splitB{0.1}{\withdrawC{\pmvM}} \mid \splitB{1}{\withdrawC{\pmvP}}}
    \big)
\end{align*}
where $\extadv{\pmvA:\pmvM}{\call{\Refund[\pmvA]}{}}$
means that only $\pmvA$ and $\pmvM$ need to agree in order 
for the contract $\Refund[\pmvA]$ to be executed, resolving the dispute.
In this case it is crucial that the renegotiation is possible 
even without the agreement between $\pmvA$ and $\pmvB$.
Indeed, if $\pmvM$ decides to refund $\pmvA$
(by authorizing $\Refund[\pmvA]$),
it is not to be expected that also $\pmvB$ agrees.
Similarly to the one discussed before, 
also this extension can be implemented on-top of Bitcoin.

\paragraph{Non-consensual renegotiation.}
In the variants of \mbox{$\adv{}$} discussed before, 
renegotiation requires one or more participants to agree.
Hence, each use of \mbox{$\adv{}$} must include 
suitable alternative branches,
to be fired in case the renegotiation fails.
In certain scenarios, we may want to renegotiate the contract 
without the participants having to agree. 
To this purpose, we can introduce a new primitive
$\ncadv{\callX{}}$, which continues as $\callX{}$ without
requiring anyone to agree.
We assume that the defining equations of this primitive have the form
$\decl{\cVarX}{\vec{\procParamA}} = \contrAdv{\valV}{\contrC}$,
where $\valV$ represents the amount of \BTC added to the contract, 
by anyone.

We exemplify the new primitive in a two-players game which starts with  
a bet of $1\BTC$ from $\pmvA$, and a bet of $2\BTC$ from $\pmvB$. 
Then, starting from $\pmvA$, players take turns adding $2\BTC$ each to the pot. 
The first one who is not able to provide the additional $2\BTC$ within a given time 
loses the game, allowing the other player to take the whole pot.
The contract is as follows:
\begin{align*}
\contrC 
  & = \setenum{ \persdep{\pmvA}{1}{\varX} \mid \persdep{\pmvB}{2}{\varY}} 
  (\ncadv{\call{\cVarX[\pmvA]}{2}} + \afterC{1}{\withdrawC{\pmvB}})
  \\
  \call{\cVarX[\pmvA]}{n}
  & = \setenum{ 2 } 
    (\ncadv{\call{\cVarX[\pmvB]}{n+1}} + \afterC{n}{\withdrawC{\pmvA}})
  \\
  \call{\cVarX[\pmvB]}{n}
  & = \setenum{ 2 } 
    (\ncadv{\call{\cVarX[\pmvA]}{n+1}} + \afterC{n}{\withdrawC{\pmvB}})
\end{align*}

Unlike \mbox{$\adv{}$}, the action $\ncadv{}$
can be fired without the authorizations of all the players: 
it just requires that the authorization to gather $2 \BTC$ 
is provided, by anyone.
Even though the sender of these $2 \BTC$ is not specified in the contract,
it is implicit in the game mechanism:
for instance, when $\call{\cVarX[\pmvA]}{n}$ calls 
$\call{\cVarX[\pmvB]}{n+1}$, only participant $\pmvB$ 
is incentivized to add $2\BTC$,
since not doing so will make $\pmvA$ win.

Implementing the $\ncadv{}$ primitive on top of Bitcoin seems unfeasible: 
even if it were possible to use complex off-chain multiparty computation 
protocols \cite{Gudgeon19iacr}, 
doing so might be impractical. 
Rather, we would like to extend Bitcoin as much as needed 
for the new primitive. 
In our implementation of BitML, 
we compile contracts to sets of transactions and make participants sign them. 
In standard BitML this is doable since, at stipulation time, 
we can finitely over-approximate the reducts of the original contract. 
Recursion can make this set infinite, 
\eg $\call{\cVarX[\pmvA]}{2}, \call{\cVarX[\pmvA]}{3}, \ldots$, 
hence impossible to compile and sign statically. 
A way to cope with this is to extend Bitcoin with \emph{malleable} 
signatures which only cover the part of the transaction not affected 
by the parameter $n$ in $\call{\cVarX[\pmvB]}{n}$. 
Further, signatures must not cover the $\txIn{}$ fields of transactions, 
since they change as recursion unfolds.
In this way, the same signature can be reused for each call. 

Adding malleability provides flexibility, but poses some risks. 
For instance, instead of redeeming the transaction corresponding to 
$\call{\cVarX[\pmvA]}{n}$ with the transaction of 
$\call{\cVarX[\pmvB]}{n+1}$
one could instead use the transaction of 
$\call{\cVarX[\pmvB]}{n+100}$,
since the two transactions have the same signature. 
To overcome this problem, we could add a new opcode 
to allow the output script of $\call{\cVarX[\pmvB]}{n}$ 
to access the parameter in the redeeming transaction, 
so to verify that it is indeed $n+1$ as intended.
Similarly, to check that we have $2\BTC$ more in the new transaction, 
an opcode could provide the value of the new output.
The same goal could be achieved by exploiting
\emph{covenants}~\cite{Moser16bw,Oconnor17bw,BZ20isola}.

\section{Conclusions}
\label{sec:conclusions}

We have investigated linguistic primitives to renegotiate BitML contracts, and 
their implementation on standard Bitcoin.
More expressive primitives could be devised
by relaxing this constraint, 
\eg assuming the extended UTXO model~\cite{Chakravarty20wtsc}.

Our verification technique is based on a sound abstraction 
of the state space of contracts.
Since this abstraction is finite-state, 
it can be model-checked to verify the required properties.
If we assume that integers are unbounded, 
and that participants always accept renegotiations,
the extension of BitML presented in Section~\ref{sec:bitml} 
can simulate a counter machine,
so making BitML Turing-complete.
Hence, any verification technique for BitML cannot be sound and complete.
Alternative techniques to model checking 
(\eg, type-based approaches~\cite{Das19arxiv})
could be used to analyse relevant contract properties.

\paragraph{Acknowledgements} 
Massimo Bartoletti is partially supported by Aut.\ Reg.\ Sardinia projects 
\textit{Sardcoin}, \textit{Smart collaborative engineering}, and
Conv.\ Fondazione di Sardegna \& Atenei Sardi project F74I19000900007 \emph{ADAM}.
Maurizio Murgia and Roberto Zunino are partially supported by MIUR PON \textit{Distributed Ledgers for Secure Open Communities}.

\bibliographystyle{alphaurl}
\bibliography{main}

\newcommand{\etalchar}[1]{$^{#1}$}
\begin{thebibliography}{ADMM14b}

\bibitem[ABC17]{ABC17post}
Nicola Atzei, Massimo Bartoletti, and Tiziana Cimoli.
\newblock A survey of attacks on {Ethereum} smart contracts {(SoK)}.
\newblock In {\em Principles of Security and Trust ({POST})}, volume 10204 of
  {\em LNCS}, pages 164--186. Springer, 2017.
\newblock \href {https://doi.org/10.1007/978-3-662-54455-6_8}
  {\path{doi:10.1007/978-3-662-54455-6_8}}.

\bibitem[ABC{\etalchar{+}}18]{bitcoinsok}
Nicola Atzei, Massimo Bartoletti, Tiziana Cimoli, Stefano Lande, and Roberto
  Zunino.
\newblock {SoK}: unraveling {Bitcoin} smart contracts.
\newblock In {\em {POST}}, volume 10804 of {\em LNCS}, pages 217--242.
  Springer, 2018.
\newblock \href {https://doi.org/10.1007/978-3-319-89722-6}
  {\path{doi:10.1007/978-3-319-89722-6}}.

\bibitem[ABL{\etalchar{+}}19]{bitmlracket}
Nicola Atzei, Massimo Bartoletti, Stefano Lande, Nobuko Yoshida, and Roberto
  Zunino.
\newblock Developing secure {Bitcoin} contracts with {BitML}.
\newblock In {\em {ESEC/FSE}}, pages 1124--1128. {ACM}, 2019.
\newblock \href {https://doi.org/https://doi.org/10.1145/3338906.3341173}
  {\path{doi:https://doi.org/10.1145/3338906.3341173}}.

\bibitem[ADMM14a]{Andrychowicz14bw}
Marcin Andrychowicz, Stefan Dziembowski, Daniel Malinowski, and Lukasz Mazurek.
\newblock Fair two-party computations via {Bitcoin} deposits.
\newblock In {\em Financial Cryptography Workshops}, volume 8438 of {\em LNCS},
  pages 105--121. Springer, 2014.
\newblock \href {https://doi.org/10.1007/978-3-662-44774-1_8}
  {\path{doi:10.1007/978-3-662-44774-1_8}}.

\bibitem[ADMM14b]{Andrychowicz14formats}
Marcin Andrychowicz, Stefan Dziembowski, Daniel Malinowski, and {\L}ukasz
  Mazurek.
\newblock Modeling {Bitcoin} contracts by timed automata.
\newblock In {\em International Conference on Formal Modeling and Analysis of
  Timed Systems ({FORMATS})}, volume 8711 of {\em LNCS}, pages 7--22. Springer,
  2014.
\newblock \href {https://doi.org/10.1007/978-3-319-10512-3_2}
  {\path{doi:10.1007/978-3-319-10512-3_2}}.

\bibitem[ADMM14c]{Andrychowicz14sp}
Marcin Andrychowicz, Stefan Dziembowski, Daniel Malinowski, and Lukasz Mazurek.
\newblock Secure multiparty computations on {Bitcoin}.
\newblock In {\em {IEEE} \mbox{S \& P}}, pages 443--458, 2014.
\newblock First appeared on Cryptology {ePrint} {Archive},
  \url{http://eprint.iacr.org/2013/784}.
\newblock \href {https://doi.org/10.1109/SP.2014.35}
  {\path{doi:10.1109/SP.2014.35}}.

\bibitem[ADMM16]{Andrychowicz16cacm}
Marcin Andrychowicz, Stefan Dziembowski, Daniel Malinowski, and Lukasz Mazurek.
\newblock Secure multiparty computations on {Bitcoin}.
\newblock {\em Commun. {ACM}}, 59(4):76--84, 2016.
\newblock \href {https://doi.org/10.1145/2896386} {\path{doi:10.1145/2896386}}.

\bibitem[And19]{btc520bytes}
Gavin Andresen.
\newblock Bitcoin script size limit, 2019.
\newblock {BIP} 16,
  \url{https://github.com/bitcoin/bips/blob/master/bip-0016.mediawiki\#520-byte-limitation-on-serialized-script-size}.

\bibitem[BBL{\etalchar{+}}21]{BB+21fc}
Massimo Bartoletti, Andrea Bracciali, Cristian Lepore, Alceste Scalas, and
  Roberto Zunino.
\newblock A formal model of {Algorand} smart contracts.
\newblock In {\em Financial Cryptography and Data Security}, volume 12674 of
  {\em LNCS}, pages 93--114. Springer, 2021.
\newblock \href {https://doi.org/10.1007/978-3-662-64322-8_5}
  {\path{doi:10.1007/978-3-662-64322-8_5}}.

\bibitem[BDM16]{Banasik16esorics}
Waclaw Banasik, Stefan Dziembowski, and Daniel Malinowski.
\newblock Efficient zero-knowledge contingent payments in cryptocurrencies
  without scripts.
\newblock In {\em {ESORICS}}, volume 9879 of {\em LNCS}, pages 261--280.
  Springer, 2016.
\newblock \href {https://doi.org/10.1007/978-3-319-45741-3_14}
  {\path{doi:10.1007/978-3-319-45741-3_14}}.

\bibitem[BK14]{Bentov14crypto}
Iddo Bentov and Ranjit Kumaresan.
\newblock How to use {Bitcoin} to design fair protocols.
\newblock In {\em {CRYPTO}}, volume 8617 of {\em LNCS}, pages 421--439.
  Springer, 2014.
\newblock \href {https://doi.org/10.1007/978-3-662-44381-1_24}
  {\path{doi:10.1007/978-3-662-44381-1_24}}.

\bibitem[BKT17]{Biryukov17wtsc}
Alex Biryukov, Dmitry Khovratovich, and Sergei Tikhomirov.
\newblock Findel: Secure derivative contracts for {Ethereum}.
\newblock In {\em Financial Cryptography Workshops}, volume 10323 of {\em
  LNCS}, pages 453--467. Springer, 2017.
\newblock \href {https://doi.org/10.1007/978-3-319-70278-0_28}
  {\path{doi:10.1007/978-3-319-70278-0_28}}.

\bibitem[BLZ20]{BZ20isola}
Massimo Bartoletti, Stefano Lande, and Roberto Zunino.
\newblock {Bitcoin} covenants unchained.
\newblock In {\em {ISoLA}}, volume 12478 of {\em ISOLA}, pages 25--42.
  Springer, 2020.
\newblock \href {https://doi.org/10.1007/978-3-030-61467-6_3}
  {\path{doi:10.1007/978-3-030-61467-6_3}}.

\bibitem[BMZ20]{BMZ20coordination}
Massimo Bartoletti, Maurizio Murgia, and Roberto Zunino.
\newblock Renegotiation and recursion in {Bitcoin} contracts.
\newblock In {\em Proc. {COORDINATION}}, volume 12134 of {\em LNCS}, pages
  261--278. Springer, 2020.
\newblock \href {https://doi.org/10.1007/978-3-030-50029-0_17}
  {\path{doi:10.1007/978-3-030-50029-0_17}}.

\bibitem[BZ17]{BZ17bw}
Massimo Bartoletti and Roberto Zunino.
\newblock Constant-deposit multiparty lotteries on {Bitcoin}.
\newblock In {\em Financial Cryptography Workshops}, volume 10323 of {\em
  LNCS}, pages 231--247. Springer, 2017.
\newblock \href {https://doi.org/10.1007/978-3-319-70278-0}
  {\path{doi:10.1007/978-3-319-70278-0}}.

\bibitem[BZ18a]{BZ18bitml}
Massimo Bartoletti and Roberto Zunino.
\newblock {BitML}: a calculus for {Bitcoin} smart contracts.
\newblock In {\em {ACM} {CCS}}, 2018.
\newblock \href {https://doi.org/10.1145/3243734.3243795}
  {\path{doi:10.1145/3243734.3243795}}.

\bibitem[BZ18b]{BZ18arxiv}
Massimo Bartoletti and Roberto Zunino.
\newblock {BitML}: a calculus for {Bitcoin} smart contracts.
\newblock {\em {IACR} Cryptol. ePrint Arch.}, page 122, 2018.
\newblock URL: \url{http://eprint.iacr.org/2018/122}.

\bibitem[BZ19]{BZ19post}
Massimo Bartoletti and Roberto Zunino.
\newblock Verifying liquidity of {Bitcoin} contracts.
\newblock In {\em {POST}}, volume 11426 of {\em LNCS}. Springer, 2019.

\bibitem[CCM{\etalchar{+}}20]{Chakravarty20wtsc}
Manuel~M.T. Chakravarty, James Chapman, Kenneth MacKenzie, Orestis Melkonian,
  Michael~Peyton Jones, and Philip Wadler.
\newblock The extended {UTXO} model.
\newblock In {\em Workshop on Trusted Smart Contracts}, 2020.

\bibitem[CDE{\etalchar{+}}02]{maude01}
Manuel Clavel, Francisco Dur{\'{a}}n, Steven Eker, Patrick Lincoln, Narciso
  Mart{\'{\i}}{-}Oliet, Jos{\'{e}} Meseguer, and Jose~F. Quesada.
\newblock Maude: specification and programming in rewriting logic.
\newblock {\em Theor. Comput. Sci.}, 285(2):187--243, 2002.
\newblock \href {https://doi.org/10.1016/S0304-3975(01)00359-0}
  {\path{doi:10.1016/S0304-3975(01)00359-0}}.

\bibitem[DBHP19]{Das19arxiv}
Ankush Das, Stephanie Balzer, Jan Hoffmann, and Frank Pfenning.
\newblock Resource-aware session types for digital contracts.
\newblock {\em CoRR}, abs/1902.06056, 2019.

\bibitem[Fla12]{Flatt12cacm}
Matthew Flatt.
\newblock Creating languages in {Racket}.
\newblock {\em Commun. {ACM}}, 55(1):48--56, 2012.
\newblock \href {https://doi.org/10.1145/2063176.2063195}
  {\path{doi:10.1145/2063176.2063195}}.

\bibitem[GMR{\etalchar{+}}19]{Gudgeon19iacr}
Lewis Gudgeon, Pedro Moreno{-}Sanchez, Stefanie Roos, Patrick McCorry, and
  Arthur Gervais.
\newblock Sok: Off the chain transactions.
\newblock {\em {IACR} Cryptology ePrint Archive}, 2019:360, 2019.

\bibitem[JES00]{PeytonJones00icfp}
Simon L.~Peyton Jones, Jean{-}Marc Eber, and Julian Seward.
\newblock Composing contracts: an adventure in financial engineering,
  functional pearl.
\newblock In {\em International Conference on Functional Programming {(ICFP)}},
  pages 280--292, 2000.
\newblock \href {https://doi.org/10.1145/351240.351267}
  {\path{doi:10.1145/351240.351267}}.

\bibitem[KB14]{Kumaresan14ccs}
Ranjit Kumaresan and Iddo Bentov.
\newblock How to use {Bitcoin} to incentivize correct computations.
\newblock In {\em {ACM} {CCS}}, pages 30--41, 2014.
\newblock \href {https://doi.org/10.1145/2660267.2660380}
  {\path{doi:10.1145/2660267.2660380}}.

\bibitem[KB16]{KumaresanB16ccs}
Ranjit Kumaresan and Iddo Bentov.
\newblock Amortizing secure computation with penalties.
\newblock In {\em {ACM} {CCS}}, pages 418--429, 2016.
\newblock \href {https://doi.org/10.1145/2976749.2978424}
  {\path{doi:10.1145/2976749.2978424}}.

\bibitem[KMB15]{Kumaresan15ccs}
Ranjit Kumaresan, Tal Moran, and Iddo Bentov.
\newblock How to use {Bitcoin} to play decentralized poker.
\newblock In {\em {ACM} {CCS}}, pages 195--206, 2015.
\newblock \href {https://doi.org/10.1145/2810103.2813712}
  {\path{doi:10.1145/2810103.2813712}}.

\bibitem[KVV16]{KumaresanVV16ccs}
Ranjit Kumaresan, Vinod Vaikuntanathan, and Prashant~Nalini Vasudevan.
\newblock Improvements to secure computation with penalties.
\newblock In {\em {ACM} {CCS}}, pages 406--417, 2016.
\newblock \href {https://doi.org/10.1145/2976749.2978421}
  {\path{doi:10.1145/2976749.2978421}}.

\bibitem[LCO{\etalchar{+}}16]{Luu16ccs}
Loi Luu, Duc-Hiep Chu, Hrishi Olickel, Prateek Saxena, and Aquinas Hobor.
\newblock Making smart contracts smarter.
\newblock In {\em {ACM} {CCS}}, pages 254--269, 2016.
\newblock \href {https://doi.org/10.1145/2976749.2978309}
  {\path{doi:10.1145/2976749.2978309}}.

\bibitem[MB17]{Miller16zerocollateral}
Andrew Miller and Iddo Bentov.
\newblock Zero-collateral lotteries in {Bitcoin} and {Ethereum}.
\newblock In {\em EuroS{\&}P Workshops}, pages 4--13, 2017.
\newblock \href {https://doi.org/10.1109/EuroSPW.2017.44}
  {\path{doi:10.1109/EuroSPW.2017.44}}.

\bibitem[MES16]{Moser16bw}
Malte M{\"o}ser, Ittay Eyal, and Emin~G{\"u}n Sirer.
\newblock {Bitcoin} covenants.
\newblock In {\em Financial Cryptography Workshops}, volume 9604 of {\em LNCS},
  pages 126--141. Springer, 2016.
\newblock \href {https://doi.org/10.1007/978-3-662-53357-4_9}
  {\path{doi:10.1007/978-3-662-53357-4_9}}.

\bibitem[OP17]{Oconnor17bw}
Russell O’Connor and Marta Piekarska.
\newblock Enhancing {Bitcoin} transactions with covenants.
\newblock In {\em Financial Cryptography Workshops}, volume 10323 of {\em
  LNCS}. Springer, 2017.
\newblock \href {https://doi.org/10.1007/978-3-319-70278-0_12}
  {\path{doi:10.1007/978-3-319-70278-0_12}}.

\bibitem[ST18]{Thompson18isola}
Pablo~Lamela Seijas and Simon~J. Thompson.
\newblock Marlowe: Financial contracts on blockchain.
\newblock In {\em {ISoLA}}, volume 11247 of {\em LNCS}, pages 356--375.
  Springer, 2018.
\newblock \href {https://doi.org/10.1007/978-3-030-03427-6_27}
  {\path{doi:10.1007/978-3-030-03427-6_27}}.

\bibitem[TDD{\etalchar{+}}18]{Tsankov18ccs}
Petar Tsankov, Andrei~Marian Dan, Dana Drachsler{-}Cohen, Arthur Gervais,
  Florian B{\"{u}}nzli, and Martin~T. Vechev.
\newblock Securify: {Practical} {Security} {Analysis} of {Smart} {Contracts}.
\newblock In {\em {ACM} {CCS}}, pages 67--82, 2018.
\newblock \href {https://doi.org/10.1145/3243734.3243780}
  {\path{doi:10.1145/3243734.3243780}}.

\end{thebibliography}

\end{document}